\documentclass[a4paper,10pt]{amsart}
\usepackage[utf8]{inputenc}
\usepackage[T1]{fontenc}
\usepackage{amsmath}
\usepackage{amsthm}
\usepackage{amssymb}
\usepackage{amsfonts}
\usepackage{mathrsfs}
\usepackage{color}
\usepackage{graphicx}
\usepackage{url}

\newcommand{\N}{\mathbb{N}}
\newcommand{\C}{\mathbb{C}}
\newcommand{\R}{\mathbb{R}}
\newcommand{\Z}{\mathbb{Z}}

\newcommand{\ena}[1]{\mathrm{e}^{#1}}

\DeclareMathOperator{\arctanh}{\mathrm{arctanh}}

\DeclareMathOperator{\sgn}{\mathrm{sgn}}

\newcommand{\dd}{\mathrm{d}}
\newcommand{\dom}{\mathrm{Dom}}
\newcommand{\ran}{\mathrm{Ran}}
\newcommand{\tr}{\mathrm{Tr}}

\newtheorem{theorem}{Theorem}[section]

\newtheorem{corollary}[theorem]{Corollary}
\newtheorem{proposition}[theorem]{Proposition}

\theoremstyle{definition}
\newtheorem{remark}[theorem]{Remark}
\newtheorem{example}[theorem]{Example}
\newtheorem{definition}[theorem]{Definition}

\numberwithin{equation}{section}

\title[Relativistic $\delta$-shell interaction supported on a straight line]{Two-dimensional Dirac operators with general $\delta$-shell interactions supported on a straight line}
\date{August 26, 2022}

\author[J. Behrndt]{Jussi Behrndt}
\address{Institut f\"{u}r Angewandte Mathematik\\
Technische Universit\"{a}t Graz\\
 Steyrergasse 30, 8010 Graz, Austria\\
E-mail: {\tt behrndt@tugraz.at}}

\author[M. Holzmann]{Markus Holzmann}
\address{Institut f\"{u}r Angewandte Mathematik\\
Technische Universit\"{a}t Graz\\
 Steyrergasse 30, 8010 Graz, Austria\\
E-mail: {\tt holzmann@math.tugraz.at}}

\author[M. Tu\v{s}ek]{Mat\v{e}j Tu\v{s}ek}
\address{Department of Mathematics, Faculty of Nuclear Sciences and Physical Engineering\\
Czech Technical University in Prague \\
Trojanova 13, 120 00, Prague \\
E-mail: {\tt matej.tusek@fjfi.cvut.cz}}

\begin{document}

\begin{abstract}
In this paper the two-dimensional Dirac operator with a general hermitian $\delta$-shell interaction supported on a straight line is introduced as a self-adjoint operator and its spectral properties are investigated in detail. In particular, it is demonstrated that the singularly continuous spectrum is always empty and that by switching a certain $\delta$-shell interaction on, it is possible to generate an eigenvalue in the gap of the spectrum of the free operator or to partially or even fully close the gap. This suggests that the studied operators may serve as interesting continuum toy-models for Dirac materials. Finally, approximations by Dirac operators with regular potentials are presented.
\end{abstract}

\maketitle

\section{Introduction}
In this paper we study the two-dimensional Dirac operator formally given by 
\begin{equation} \label{eq:formal}
\hat{H}_{\eta,\tau,\lambda,\omega} := \sigma_1(-i\partial_x)+\sigma_2(-i\partial_y)+\sigma_3 m+(\sigma_0 \eta+\sigma_3\tau+(\sigma \cdot {\bf t}) \lambda+(\sigma \cdot {\bf n}) \omega)\delta_\Sigma,
\end{equation}
where $m,\eta,\tau,\lambda,\omega\in\R$, $\sigma_{i}$ with $i=1,2,3$ are the usual Pauli matrices defined in~\eqref{def_Pauli} below, $\sigma_0$ is the $2\times 2$ identity matrix, $\delta_\Sigma$ is the Dirac $\delta$-distribution supported on the straight line $\Sigma:=\{(0,y)|~ y\in\R\}$, and ${\bf t}$ and ${\bf n}$ are the tangential and normal vector along $\Sigma$, respectively; the notation $\sigma \cdot x = \sigma_1 x_1 + \sigma_2 x_2$ was used for a vector $x=(x_1,x_2)$. The singular term is called the $\delta$-shell potential. Since $\{\sigma_i\}_{i=0}^3$ constitutes a basis of the space of $2\times 2$ hermitian matrices, we deal with a general hermitian $\delta$-shell interaction. The first two interaction terms 
$\sigma_0 \eta$ and $\sigma_3\tau$
in \eqref{eq:formal} are referred to as the electrostatic and the Lorentz scalar interaction, respectively. The two remaining terms $(\sigma \cdot {\bf t}) \lambda$ and $(\sigma \cdot {\bf n}) \omega$ were recently related to the magnetic interaction \cite{CaLoMaTu_21}. The expression in~\eqref{eq:formal} is rigorously defined as a differential operator in an $L^2$-space with the help of transmission conditions for the functions in the operator domain.

The operator in~\eqref{eq:formal} may serve as a toy-model for graphene and similar Dirac materials. In fact, Dirac operators with scaled magnetic potentials that converge in the distributional sense to $(\sigma \cdot {\bf t}) \lambda \delta_\Sigma$ have been studied before as models of graphene with  very localized magnetic barriers, see for example \cite{FoGuKa_08,MaVaPe_09,PeNe_09}.
Another application is related to the fact that there may appear defects in the synthesised graphene, grain boundaries being the most common ones \cite{LeKiBr_12}. For a line defect, the corresponding continuum model was derived in \cite{RoPeLo_12}, see also \cite{RoPeLo_13,Ro_16}.  It coincides with the model studied in the current paper in the sense that the transmission condition deduced there is the same as the transmission condition for functions in the domain of $\hat{H}_{\eta,\tau,\lambda,\omega}$. 
Let us stress that most of the cited papers are concerned with scattering issues. In this paper we want to study, for the first time, the self-adjointness and the spectral properties of the two-dimensional Dirac operator with a general $\delta$-shell interaction supported on a straight line.

Concerning the mathematical study of the operator in~\eqref{eq:formal}, it turns out that $\hat{H}_{\eta,\tau,\lambda,\omega}$ is always unitarily equivalent to an operator of the same type but with $\omega=0$; cf. Section~\ref{sec:unitary_eq} and \cite{CaLoMaTu_21}. Therefore, it is sufficient to study a three-parametric family of operators, which will be denoted by $\hat{H}_{\eta,\tau,\lambda} := \hat{H}_{\eta,\tau,\lambda, 0}$. The operator $\hat{H}_{\eta,\tau,\lambda}$ was originally studied for $\Sigma$ being a smooth closed non-self-intersecting curve and $\lambda=0$ in \cite{BHOP20}. Subsequently, a general hermitian $\delta$-shell interaction supported on the closed curve with possibly non-constant coupling constants was introduced and analysed in \cite{CaLoMaTu_21}. In particular, $\hat{H}_{\eta,\tau,\lambda}$ was proven to be self-adjoint in $L^2(\R^2;\C^2)$ when $\lambda=0$ or if 
\begin{equation} \label{eq:criticality}
 \left(\frac{\eta^2-\tau^2-\lambda^2}{4}-1\right)^2-\lambda^2\neq 0
\end{equation}
using advanced boundary triple techniques.
We say that we are in the non-critical case if~\eqref{eq:criticality} holds true. In the opposite (critical) case, it is known that the functions in the domain of $\hat{H}_{\eta,\tau,\lambda}$ are less regular than in the non-critical case. Note that for $\lambda=0$, the criticality condition yields $\eta^2-\tau^2=4$ and, in the complementary case when $\eta=\tau=0$, it yields $\lambda=\pm 2$. The case $(\eta,\tau,\lambda)=(0,0,\pm 2)$ is both critical and confining, where the latter means that $\hat{H}_{\eta,\tau,\lambda}$ decouples into a direct sum of two operators acting on functions defined on the open domains in $\R^2$ that are separated by $\Sigma$. This happens if and only if $\eta^2-\tau^2-\lambda^2=-4.$
From a physical point of view, confinement means that the $\delta$-potential is impenetrable for particles.
The self-adjointness of $\hat{H}_{0,0,\pm 2}$ was shown differently in \cite{Sch_95} employing the supersymmetric structure of the operator.

Let us briefly summarize the spectral properties of $\hat{H}_{\eta,\tau,\lambda}$ that were obtained in \cite{BHOP20,CaLoMaTu_21} when $\Sigma$ is a closed curve. In the non-critical case,  the essential spectrum is given by
$$\sigma_{\textup{ess}}(\hat{H}_{\eta,\tau,\lambda})=(-\infty,-|m|]\cup[|m|,+\infty)$$
and the discrete spectrum is finite. If $\lambda=0$ and $\eta^2-\tau^2=4$ then $$\sigma_{\textup{ess}}(\hat{H}_{\eta,\tau,0})=(-\infty,-|m|]\cup\left\{-\frac{\tau}{\eta}m\right\}\cup[|m|,+\infty).$$ Finally, for $(\eta,\tau,\lambda)=(0,0,\pm 2)$, $\sigma(\hat{H}_{0,0,\pm 2})=(-\infty,-|m|]\cup[|m|,+\infty)$, $\pm m$ are eigenvalues of infinite multiplicity, and there is a sequence of embedded eigenvalues $\{\pm\sqrt{m^2+\lambda_k}\}_{k\in\N}$, where the $\lambda_k$'s are the eigenvalues of the Dirichlet Laplacian on the bounded domain enclosed by $\Sigma$.

If $\Sigma$ is an unbounded curve, it is not clear how to prove the self-adjointness of $\hat{H}_{\eta,\tau,\lambda}$. We view the current model as an explicitly solvable example of what spectral behaviour may be expected for the two-dimensional Dirac operator with a $\delta$-shell interaction supported on an unbounded curve. So far, there are only the papers \cite{R21, R22} dealing with the non-critical case under quite general assumptions that contain less explicit results, and the publication \cite{BeHoTu_21} by the authors of the present paper for the case when $\Sigma$ is a straight line and $\tau=\lambda=0$. 
The spectrum of $\hat{H}_{\eta,0,0}$ may be now very different from the case when $\Sigma$ was a closed curve.  In the non-critical case, the gap $(-|m|,|m|)$ is partially closed by either negative or positive energies and the spectrum is purely continuous, see \eqref{eq:spec_es} or \cite[Theorem~1.1]{BeHoTu_21}. In the critical case, i.e., for $\eta=\pm 2$, 
$$\sigma(\hat{H}_{\pm 2,0,0})=(-\infty,-|m|]\cup\{0\}\cup[|m|,+\infty),$$ 
where the point $0$ is an isolated eigenvalue of infinite multiplicity and the rest of the spectrum is purely continuous. Hence, when the interaction strength changes from a non-critical to a  critical interaction strength, then a part of the continuous spectrum collapses into an eigenvalue of infinite multiplicity.

In the present paper, we will employ the partial Fourier transform and work with the direct integral of a family of one-dimensional Dirac operators. 
In more detail, we will start with the free two-dimensional Dirac operator, whose spectrum is purely absolutely continuous and equal to $(-\infty,-|m|]\cup[|m|,+\infty)$, and decompose it into the direct integral of  the following fiber operators
\begin{equation} \label{eq:fib_op}
H[k]=\sigma_1(-i\partial_x)+\sigma_2k+\sigma_3 m,
\end{equation}
which are nothing but the one-dimensional Dirac operators perturbed by the hermitian term $\sigma_2 k$, where $k\in\R$ is the momentum in the $y$-direction. Then we add the standard one-dimensional point interaction at $x=0$ to every $H[k]$. The associated operators are known to be self-adjoint. Finally, we introduce the two-dimensional Dirac operator with a $\delta$-shell interaction supported on the line $\Sigma$ as the direct integral of these operators. The resulting operator is self-adjoint by construction. Moreover, its spectral properties follow from those of its fibers. The spectra of the one-dimensional Dirac operators with point interactions have been studied before in \cite{BeDa_94}, but now we have to include an additional perturbation, namely  $\sigma_2 k$. We will conclude that there are at most two discrete eigenvalues in the spectrum of each of the one-dimensional fibers and the rest of its spectrum is purely absolutely continuous, see Theorem \ref{theo:spectra_fibers}. Using an abstract criterion for the absence of singular continuous spectrum of direct integrals of self-adjoint operators
(see Appendix~\ref{section_appendix}) we  will show that the spectrum of the full operator $\hat{H}_{\eta,\tau,\lambda}$ may be very different from both the spectrum of the free operator and the spectrum of the Dirac operator with a $\delta$-shell interaction supported on a closed curve. For certain special values of the coupling parameters, there is  an isolated or embedded eigenvalue of infinite multiplicity in the spectrum of $\hat{H}_{\eta,\tau,\lambda}$ and the rest of the spectrum is purely absolutely continuous. For the remaining values of the coupling parameters, the spectrum is purely absolutely continuous. The gap $(-|m|,|m|)$ in the spectrum of the free operator may be fully or partially closed; cf. Theorem \ref{theo:full_spectrum} for details.
A part of the spectrum of $\hat{H}_{\eta,\tau,\lambda}$ is due to the energy bands $z:\,k\mapsto z(k)$, where $z(k)$ are eigenvalues of the corresponding fiber operators. Remarkably, for certain values of the coupling parameters these bands are linear. If $z=z(k)$ is  constant then $z$ is an eigenvalue of infinite multiplicity in the spectrum of $\hat{H}_{\eta,\tau,\lambda}$. If $z$ is linear but non-constant then there exist a sort of edge states that are localized close to $\Sigma$ and propagate in $y$-direction with constant speed $\dd z/\dd k$, see Section \ref{sec:LinBands}.

The singular potential term in \eqref{eq:formal} (with $\omega=0$) is realized by a transmission condition along the line $\Sigma$ in the definition of $\hat{H}_{\eta,\tau,\lambda}$. To understand the nature of the delta-shell interaction it is tempting to approximate $\delta_\Sigma$ by a family $(W_\varepsilon)_{\varepsilon>0}$ of scaled regular (and even bounded) potentials that converge to $\delta_\Sigma$ in the sense of distributions. This way, one gets a family of operators 
\begin{equation*} 
\hat{H}^\varepsilon_{\eta,\tau,\lambda}:=\hat H+(\sigma_0 \eta+\sigma_3\tau+\sigma_2\lambda)W_\varepsilon,
\end{equation*}
where $\hat H$ is the free two-dimensional Dirac operator. Surprisingly, $\hat{H}^\varepsilon_{\eta,\tau,\lambda}$ does not converge to $\hat{H}_{\eta,\tau,\lambda}$ as $\varepsilon\to 0$. Instead, we will show that the strong resolvent limit of $\hat{H}^\varepsilon_{\eta,\tau,\lambda}$ is $\hat{H}_{C\eta,C\tau,C\lambda}$, where  $C\in\R$ depends non-trivially on $\eta,\,\tau,$ and $\lambda$, see Proposition \ref{proposition_approximation}. The same renormalization of the coupling constants has been observed before in different settings \cite{CaLoMaTu_21,Hu_99,MaPi_17,MaPi_18,Se_89,Tu_20}. Note that if one starts in the one-dimensional situation with non-local potentials, e.g., projections on $W_\varepsilon$, instead of the local ones, then the renormalization of the coupling constants in the limit $\varepsilon\to 0$ does not occur \cite{Heriban,Se_89}.

The paper is organized as follows. In Section \ref{sec:free_op}, we  recall basic properties of the free two-dimensional Dirac operator and introduce the fiber operators \eqref{eq:fib_op} that appear in its direct integral decomposition. In Section \ref{sec:pointInt}, we add a general hermitian point interaction  to these fiber operators. Then we  show in Section \ref{sec:unitary_eq} that without loss of generality we may study just a three-parametric subfamily of point interactions. The spectra of the fiber operators with additional point interactions are analyzed in Section \ref{sec:spec_fibers}. 
The two-dimensional Dirac operator with a $\delta$-shell interaction supported on a straight line is introduced and its spectral properties are studied in Section \ref{sec:2dOp}.
Section \ref{sec:approx} is devoted to approximations of $\delta$-shell interactions by regular potentials. Finally, in Appendix~\ref{section_appendix} we briefly recall some properties 
of direct integrals of self-adjoint operators and we prove a criterion on the absence of singular continuous spectrum.

\section{Two-dimensional free Dirac ope\-rator} \label{sec:free_op}
In this section we recall the definition and some basic properties of the two-dimensional free Dirac operator $\hat H$. Let 
\begin{equation} \label{def_Pauli}
 \sigma_1=\begin{pmatrix} 0 & 1 \\ 1 & 0 \end{pmatrix},\quad 
 \sigma_2=\begin{pmatrix} 0 & -i \\ i & 0 \end{pmatrix},\quad\text{and}\quad
 \sigma_3=\begin{pmatrix} 1 & 0 \\ 0 & -1 \end{pmatrix}\quad 
\end{equation}
be the Pauli matrices, denote by $\sigma_0$ the $2\times 2$ identity matrix, and let $m\in\R$. Then 
\begin{equation}\label{hath}
\begin{split}
  \hat H&=\sigma_1(-i\partial_x)+\sigma_2(-i\partial_y)+\sigma_3 m\equiv-i\sigma\cdot\nabla+\sigma_3 m,\\
 \dom(\hat H)&=H^1(\R^2;\C^2),
\end{split}
\end{equation}
is self-adjoint in the Hilbert space $L^2(\R^2;\C^2)$; here and in the following $H^1$ denotes the Sobolev space of $\C^2$-valued $L^2$-functions with first order weak derivatives in $L^2$. After the partial Fourier--Plancherel  transform in the $y$-variable $\mathscr{F}_{y\to k}$ the operator $\hat H$ in \eqref{hath} decomposes into the direct integral,
$$\mathscr{F}_{y\to k}\hat H\mathscr{F}_{y\to k}^{-1}=\int_\R^\oplus H[k]\dd k,$$
where each fiber acts as
\begin{equation}\label{hkop}
\begin{split}
 H[k]&=\sigma_1(-i\partial_x)+\sigma_2 k+\sigma_3 m=\begin{pmatrix}
                                                    m & -i(\partial_x+k)\\
                                                    -i(\partial_x-k) & -m
                                                   \end{pmatrix},\\
                          \dom(H[k])&=H^1(\R;\C^2),                        
\end{split}
\end{equation}
and we identify $L^2(\R^2;\C^2)=\int_\R^\oplus L^2(\R;\C^2)\dd k$;
cf. Appendix~\ref{section_appendix}. With the help of the Fourier transform in the $x$-variable one shows that the fiber operators $H[k]$, $k\in\R$, are unitarily equivalent to the multiplication operator associated with the matrix-valued function
\begin{equation} \label{eq:mult_matrix}
\begin{pmatrix}
m & p-ik\\
p+ik & -m
\end{pmatrix}
\end{equation}
in $L^2(\R,\dd p;\C^2)$. Consequently, the operators $H[k]$ are
self-adjoint in $L^2(\R;\C^2)$ and, since the eigenvalues of the matrix \eqref{eq:mult_matrix} are $\pm \sqrt{m^2+k^2+p^2}$, we have 
\begin{equation}\label{spechk}
\sigma(H[k])=\sigma_{\textup{ac}}(H[k])=\bigl(-\infty,-\sqrt{m^2+k^2}\bigr]\cup\bigl[\sqrt{m^2+k^2},+\infty\bigr);
\end{equation}
see \cite{PaRi_14} and \cite[Theorem 1.1]{thaller} for similar considerations when $k=0$ and in the three-dimensional case, respectively.
A direct calculation shows that
\begin{equation*}
  (H[k])^2 = \sigma_0(-\partial_{xx} + m^2 + k^2), \quad \dom((H[k])^2)=H^2(\R;\C^2).
\end{equation*}
Hence, we get for $z \notin \sigma(H[k])$ and $f \in L^2(\mathbb{R}; \mathbb{C}^2)$
\begin{equation*}
\begin{split}
  (H[k] - z)^{-1} f &= (H[k] + z) \big(\sigma_0(-\partial_{xx} + m^2 + k^2 -z^2)\big)^{-1} f \\ 
  &= \int_\mathbb{R} G_z(\cdot-y) f(y) \textup{d} y
\end{split}
  \end{equation*}
with
\begin{equation} \label{def_G_z}
  G_z(x) = \frac{i}{2 \xi_k(z)} \ena{i \xi_k(z) |x|} \begin{pmatrix} z+m & \xi_k(z) \sgn(x) - i k\\ \xi_k(z) \sgn(x) + ik & z-m \end{pmatrix},
\end{equation}
and $\xi_k(z) := \sqrt{z^2-k^2-m^2}$, where the complex square root is chosen such that $\textup{Im} \sqrt{w} > 0$ for $w \in \mathbb{C} \setminus [0,+\infty)$.
Here we used the fact that the integral kernel of $(-\partial_{xx}-w)^{-1}$ is $$\frac{i}{2\sqrt{w}}\ena{i\sqrt{w}|x-y|}.$$

\section{Adding a point interaction to the fiber operators}  \label{sec:pointInt}
In this section we consider the formal first order differential expression
\begin{equation} \label{eq:formalFree}
D[k]=\sigma_1(-i\partial_x)+\sigma_2 k+\sigma_3 m
\end{equation}
and the perturbed differential expression
\begin{equation} \label{eq:formalOp}
D[k]_{\eta,\tau,\lambda,\omega}=D[k]+(\sigma_0 \eta +\sigma_3 \tau +\sigma_2 \lambda +\sigma_1 \omega)\delta,
\end{equation}
where $\eta,\tau,\lambda,\omega\in\R$, and $\delta$ is the Dirac $\delta$-distribution supported at $x=0$. Our aim is to associate 
a self-adjoint operator $H[k]_{\eta,\tau,\lambda,\omega}$ in $L^2(\R;\C^2)$ to \eqref{eq:formalOp} using appropriate transmission conditions at $x=0$ for the functions in the 
operator domain. Note that for $\eta=\tau=\lambda=\omega=0$ the expression $D[k]_{\eta,\tau,\lambda,\omega}$ reduces to $D[k]$ in \eqref{eq:formalFree}, where the corresponding 
self-adjoint operator $H[k]$ is given in \eqref{hkop}.

Consider
$\psi\equiv\psi_-\oplus\psi_+\in H^1(\R_-;\C^2)\oplus H^1(\R_+;\C^2)$ and denote the traces of $\psi_-$ and $\psi_+$ at $x=0$ by $\psi(0_-)$ and $\psi(0_+)$, respectively. 
Recall that the traces coincide with the boundary values of the continuous representatives of $\psi_-$ and $\psi_+$, respectively. 
Integration by parts shows that the distribution $D[k]\psi$ acts on $\varphi \in C_0^\infty(\mathbb{R};\C^2)$ as 
\begin{equation*}
 \begin{split}
  \big( &D[k]\psi, \varphi \big) = \int_\mathbb{R} \psi  (\sigma_1(i\partial_x)-\sigma_2 k+\sigma_3 m) \varphi\, \textup{d} x \\
  &=\int_{-\infty}^0 (\sigma_1(-i\partial_x)+\sigma_2 k+\sigma_3 m) \psi_ - \,\varphi\,\textup{d} x+ \int_0^\infty (\sigma_1(-i\partial_x)+\sigma_2 k+\sigma_3 m)\psi_+ \,\varphi\,\textup{d} x \\
  &\qquad \qquad -i \sigma_1(\psi(0_+)-\psi(0_-)) \varphi(0);
 \end{split}
\end{equation*}
here $(\cdot, \cdot)$ denotes the duality product in $(C_0^\infty(\mathbb{R};\C^2))' \times C_0^\infty(\mathbb{R};\C^2)$.
Therefore, in the sense of distributions we conclude
\begin{equation} \label{distr_repr}
D[k] \psi=D[k]\psi_-\oplus D[k]\psi_+ -i \sigma_1(\psi(0_+)-\psi(0_-))\delta.
\end{equation}
We define the product of the (not necessarily smooth) function $\psi$ with the $\delta$-distribution as 
\begin{equation}\label{dede}
(\psi\delta,\varphi):=\frac{\psi(0_+)+\psi(0_-)}{2}\,\varphi(0).
\end{equation}
Now we want $D[k]_{\eta,\tau,\lambda,\omega}\psi$ to be generated by an element in $L^2(\R;\C^2)$, and hence the singular contributions in 
\begin{equation*}
 D[k]_{\eta,\tau,\lambda,\omega}\psi=D[k]\psi +(\sigma_0 \eta +\sigma_3 \tau +\sigma_2 \lambda +\sigma_1 \omega)\psi\delta
\end{equation*}
and~\eqref{distr_repr} have to cancel out.
Using \eqref{dede} this yields
$$-i\sigma_1 (\psi(0_+)-\psi(0_-))+(\sigma_0 \eta +\sigma_3 \tau +\sigma_2 \lambda +\sigma_1 \omega)\frac{\psi(0_+)+\psi(0_-)}{2}=0,$$ 
and it  is convenient to rewrite this as
\begin{equation} \label{eq:TC}
(2i\sigma_1-M_{\eta,\tau,\lambda,\omega})\psi(0_+)=(2i\sigma_1+M_{\eta,\tau,\lambda,\omega})\psi(0_-)
\end{equation}
with 
\begin{equation*}
M_{\eta,\tau,\lambda,\omega}:=\sigma_0 \eta +\sigma_3 \tau +\sigma_2 \lambda +\sigma_1 \omega=
\begin{pmatrix} \eta+\tau & -i\lambda+\omega \\ i\lambda+\omega & \eta-\tau\end{pmatrix}.
\end{equation*}

The above considerations lead to the following definition of the perturbed operator $H[k]_{\eta,\tau,\lambda,\omega}$.

\begin{definition}\label{popo}
Let $\eta,\tau,\lambda,\omega\in\R$ and $D[k]$, $k\in\R$, be as in \eqref{eq:formalFree}. 
The operator $H[k]_{\eta,\tau,\lambda,\omega}$ associated with the perturbed differential expression $D[k]_{\eta,\tau,\lambda,\omega}$ in
$L^2(\R;\C^2)$ is defined by
\begin{align*}
 H[k]_{\eta,\tau,\lambda,\omega}\psi &=D[k]\psi_-\oplus D[k]\psi_+\\
 &=(\sigma_1(-i\partial_x)+\sigma_2 k+\sigma_3 m)\psi_-\oplus (\sigma_1(-i\partial_x)+\sigma_2 k+\sigma_3 m)\psi_+,\\
\dom(H[k]_{\eta,\tau,\lambda,\omega})&=\bigl\{\psi=\psi_-\oplus\psi_+\in H^1(\R_-;\C^2)\oplus H^1(\R_+;\C^2)|\, \eqref{eq:TC}\text{ holds}\bigr\}.
 \end{align*}
\end{definition}
 
In the following we shall examine the transmission conditions \eqref{eq:TC} of the functions in $\dom(H[k]_{\eta,\tau,\lambda,\omega})$. It is convenient to 
distinguish the two possible cases
$\det(2i\sigma_1-M_{\eta,\tau,\lambda,\omega})\neq 0$ and $\det(2i\sigma_1-M_{\eta,\tau,\lambda,\omega})=0$. In the first case  
we have 
$$\det(2i\sigma_1-M_{\eta,\tau,\lambda,\omega})=d-(2i-\omega)^2\neq 0,$$ where 
\begin{equation*} 
d:=\eta^2-\tau^2-\lambda^2,
\end{equation*}
and hence we can further rewrite \eqref{eq:TC} in the form 
\begin{equation} \label{eq:TClambda}
\psi(0_+)=\Lambda\psi(0_-) 
\end{equation}
with 
\begin{equation} \label{eq:LambdaDef}
\begin{split}
\Lambda&=(2i\sigma_1-M_{\eta,\tau,\lambda,\omega})^{-1}(2i\sigma_1+M_{\eta,\tau,\lambda,\omega})\\
&=\frac{1}{d-(2i-\omega)^2}\begin{pmatrix}
4+4\lambda+\omega^2-d &4i(\tau-\eta)\\
-4i(\tau+\eta)& 4-4\lambda+\omega^2-d
\end{pmatrix}.
\end{split}
\end{equation}
Note that 
$$
\det{\Lambda}=\frac{1}{d-(2i-\omega)^2}(d-(2i+\omega)^2)
$$
and hence $|\det{\Lambda}|=1$. Moreover, $\Lambda$ is a multiple of a matrix with real diagonal terms and purely imaginary off-diagonal terms. Therefore, the domain of $H[k]_{\eta,\tau,\lambda,\omega}$ is identical to the domain of one of the self-adjoint extensions of  $H[0]$ restricted to the functions that vanish at $x=0$, that were studied in \cite{BeDa_94}  (note that the transmission condition equivalent to \eqref{eq:TClambda} appeared before in \cite{DiExSe_89} and some special cases were investigated even earlier in \cite{GeSe_87}).  The action of these extensions is the same as the action of $H[k]_{\eta,\tau,\lambda,\omega}$ up to the term $(-\sigma_2 k)$, which is bounded and symmetric. Using the Kato-Rellich theorem, we conclude that $H[k]_{\eta,\tau,\lambda,\omega}$ is also self-adjoint. 

\begin{remark}
Although for every hermitian matrix $M_{\eta,\tau,\lambda,\omega}$ such that $\det(2i\sigma_1-M_{\eta,\tau,\lambda,\omega})\neq 0$, $\Lambda$ given by \eqref{eq:LambdaDef} is always the matrix describing the transmission condition of a certain self-adjoint extension studied in \cite{BeDa_94}, the converse is not true. More precisely, there are admissible $\Lambda$'s (e.g., $\Lambda=-\sigma_0$) that are not generated by any $M_{\eta,\tau,\lambda,\omega}$.
\end{remark}

Let us now consider the second case $\det(2i\sigma_1-M_{\eta,\tau,\lambda,\omega})=0$. This happens if and only if $\omega=0$ and $d=-4$. Multiplying \eqref{eq:TC}  by $\sigma_1$ we get
$$(2i\sigma_0-\sigma_1 M_{\eta,\tau,\lambda,0})\psi(0_+)=(2i\sigma_0+\sigma_1 M_{\eta,\tau,\lambda,0})\psi(0_-).$$
Multiplying by $(2i\sigma_0\pm\sigma_1 M_{\eta,\tau,\lambda,0})$ on both sides of this identity and employing the fact that 
\begin{equation} \label{eq:Msquared}
(\sigma_1 M_{\eta,\tau,\lambda,0})^2=d=-4,
\end{equation}
we obtain
$$(-8\pm 4i\sigma_1 M_{\eta,\tau,\lambda,0})\psi(0_\mp)=0,$$
which is equivalent to
\begin{equation} \label{eq:TCdecoupled}
(2i\sigma_1\pm M_{\eta,\tau,\lambda,0})\psi(0_\mp)=0.
\end{equation}
On the other hand, functions obeying \eqref{eq:TCdecoupled} clearly satisfy \eqref{eq:TC}. Hence, in the present case we may rewrite the domain of the operator 
$H[k]_{\eta,\tau,\lambda,0}$ in Definition~\ref{popo} in the form
\begin{equation*}
\dom(H[k]_{\eta,\tau,\lambda,0})=\bigl\{\psi\equiv\psi_-\oplus\psi_+\in H^1(\R_-;\C^2)\oplus H^1(\R_+;\C^2)|\, \eqref{eq:TCdecoupled}\text{ holds}\bigr\}.
\end{equation*}
The operator $H[k]_{\eta,\tau,\lambda,0}$ (with $d=-4$) decouples into a direct sum of operators acting separately on $L^2(\R_-;\C^2)$ and $L^2(\R_+;\C^2)$. Let us look closer at the boundary conditions for $\psi_\mp$. Note that the rank of the matrix in \eqref{eq:TCdecoupled} is one and, except for the case $\tau=\eta$ and $\lambda=\mp 2$, the condition \eqref{eq:TCdecoupled} for $\psi_\mp(0)$ is equivalent to
\begin{equation*}
(\tau-\eta)\psi_\mp^2(0)=i(\lambda\pm 2)\psi_\mp^1(0),
\end{equation*}
where the upper index denotes the components of the $\mathbb{C}^2$-valued function $\psi$. If $\tau=\eta$ together with $\lambda=\mp2$ we may rewrite \eqref{eq:TCdecoupled} for $\psi_\mp$ as $\psi_-^2(0)=\frac{i\eta}{2}\psi_-^1(0)$ or $\psi_+^2(0)=-\frac{i\eta}{2}\psi_+^1(0)$, respectively.
Therefore, in the case $\omega=0$ and $d=-4$ the boundary condition for each of the subsystems is of the form $\psi_\pm^2(0)=i\zeta_\pm\psi_\pm^1(0)$, where $\zeta_\pm\in \R\cup\{+\infty\}$ (the choice $\zeta_\pm=+\infty$ is understood as $\psi_\pm^1(0)=0$). It has been already established that the corresponding decoupled operators are self-adjoint on $L^2(\R_\pm;\C^2)$ and their spectrum as a function of $k$ was studied in detail \cite{GrLe_06}. This special case will not be further analyzed in the current manuscript.

For the convenience of the reader we summarize the findings of this section in the following theorem:

\begin{theorem}
For $\eta,\tau,\lambda,\omega\in \R$ and all $k\in\R$ the operator $H[k]_{\eta,\tau,\lambda,\omega}$ 
in Definition~\ref{popo}  is self-adjoint in $L^2(\R;\C^2)$. Furthermore, the following holds:
\begin{itemize}
\item [{\rm (i)}] If $\omega\neq 0$ or $d=\eta^2-\tau^2-\lambda^2\neq -4$, then 
\begin{equation*}
\dom(H[k]_{\eta,\tau,\lambda,\omega})=\bigl\{\psi=\psi_-\oplus\psi_+\in H^1(\R_-;\C^2)\oplus H^1(\R_+;\C^2)|\, \psi(0_+)=\Lambda\psi(0_-) \bigr\},
\end{equation*}
where $\Lambda$ is the $2\times 2$-matrix in \eqref{eq:LambdaDef}.
\item [{\rm (ii)}]  If $\omega=0$ and $d=-4$, then  $H[k]_{\eta,\tau,\lambda,0}$ decouples in the orthogonal sum of two self-adjoint operators
$H[k]_{\eta,\tau,\lambda,0}^\pm$ in $L^2(\R^\pm;\C^2)$, that is, 
$$H[k]_{\eta,\tau,\lambda,0}=H[k]_{\eta,\tau,\lambda,0}^-\oplus H[k]_{\eta,\tau,\lambda,0}^+,$$
where 
$$H[k]_{\eta,\tau,\lambda,0}^\pm\varphi_\pm=(\sigma_1(-i\partial_x)+\sigma_2 k+\sigma_3 m)\varphi_\pm$$
and 
$$\dom(H[k]_{\eta,\tau,\lambda,0}^\pm)=\bigl\{\varphi = (\varphi^1, \varphi^2)  \in H^1(\R_\pm;\C^2)|\,
(\tau-\eta)\varphi^2(0)=i(\lambda\mp 2)\varphi^1(0)\bigr\}$$
if $\tau\neq\eta$ or $\lambda\neq\pm 2$, and 
\begin{align*}
&\dom(H[k]_{\eta,\eta,-2,0}^-)=\bigl\{\varphi \in H^1(\R_-;\C^2)|\, \varphi^2(0)=\tfrac{i\eta}{2} \varphi^1(0)\bigr\},\\
&\dom(H[k]_{\eta,\eta,2,0}^+)=\bigl\{\varphi \in H^1(\R_+;\C^2)|\, \varphi^2(0)=-\tfrac{i\eta}{2}\varphi^1(0)\bigr\}.
\end{align*}
\end{itemize}
\end{theorem}

We conclude this section by providing an explicit formula for the resolvent of $H[k]_{\eta,\tau,\lambda,\omega}$.
Let $z \in \mathbb{C} \setminus \mathbb{R}$. Recall that \emph{the Green function} $G_z$ is defined by~\eqref{def_G_z} and introduce the matrix
\begin{equation*}
  C_z := \frac{i}{2 \sqrt{z^2-k^2-m^2}} \begin{pmatrix} z+m & -ik \\ ik & z-m \end{pmatrix}
  =\frac{i}{2 \xi_k(z)} \begin{pmatrix} z+m & -ik \\ ik & z-m \end{pmatrix}.
\end{equation*}
Then one can check by a direct calculation that the matrix $\sigma_0 + M_{\eta,\tau,\lambda,\omega} C_z$ is invertible and that 
\begin{equation} \label{krein_formula}
 \begin{split}
 & (H[k]_{\eta,\tau,\lambda,\omega} - z)^{-1} f =\\
 &\qquad (H[k]-z)^{-1}f - G_z(\cdot) \big( \sigma_0 + M_{\eta,\tau,\lambda,\omega} C_z \big)^{-1} M_{\eta,\tau,\lambda,\omega} \int_\mathbb{R} G_z(-y) f(y) \textup{d} y
 \end{split}
\end{equation}
holds for all $f \in L^2(\mathbb{R}; \mathbb{C}^2)$. Indeed, if one defines the right hand side of this equation as $g$, then $g \in H^1(\mathbb{R}_-; \mathbb{C}^2) \oplus H^1(\mathbb{R}_+; \mathbb{C}^2)$ and this function satisfies the transmission conditions in~\eqref{eq:TC}, i.e., $g \in \dom(H[k]_{\eta,\tau,\lambda,\omega})$. Moreover, as $G_z$ is the fundamental solution for $\sigma_1 (-i \partial_x )+ \sigma_2 k+ \sigma_3 m-  \sigma_0 z$, we get $(H[k]_{\eta,\tau,\lambda,\omega} - z) g = f$. Putting this together with the fact that $(H[k]_{\eta,\tau,\lambda,\omega}-z)$ is injective, we infer that the resolvent formula 
\eqref{krein_formula} is true. For similar formulae as~\eqref{krein_formula} in the case $k=0$ we refer to \cite[eq.~(2.4)]{PaRi_14} and \cite[eq.~(19)]{BeDa_94}; this kind of ansatz 
and expression for the resolvent of $H[k]_{\eta,\tau,\lambda,\omega}$ is often referred to as a so-called Krein formula; cf. \cite[\S 106]{AkGl}.

\section{Unitary equivalences} \label{sec:unitary_eq}
In the same vein as in \cite{CaLoMaTu_21} the interaction term $\sigma_1 \omega \delta$ may be "gauged away". More concretely, we will construct a unitary transform in $L^2(\R;\C^2)$ such that $H[k]_{\eta,\tau,\lambda,\omega}$ is unitarily equivalent to $H[k]_ {X\eta,X\tau,X\lambda,0}$ with some $X\in\R$. Let $\chi_M$ stand for the indicator function of a set $M$. Since $H[k]_{\eta,\tau,\lambda,\omega}$ and $H[k]_ {X\eta,X\tau,X\lambda,0}$ differ only in the transmission conditions at $x=0$, it is quite natural to start with the following ansatz for the unitary transform:
\begin{equation} \label{eq:Udef}
U_z\varphi:=\chi_{\R_+}\varphi+z\chi_{\R_-}\varphi \quad\text{with}\quad z\in\C,\, |z|=1,\qquad \varphi\in L^2(\R;\C^2).
\end{equation}
We would like to find $z\in\C$ with $|z|=1$ and $X\in\R$ such that
\begin{equation} \label{eq:unitEq}
H[k]_{\eta,\tau,\lambda,\omega}=U_z H[k]_ {X\eta,X\tau,X\lambda,0} U_{\bar z}.
\end{equation}
Note that \eqref{eq:unitEq} is valid if and only if the following equivalence holds:
\begin{equation*}
 \psi \in \dom(H[k]_{\eta,\tau,\lambda,\omega}) \quad\text{if and only if}\quad U_{\bar z}\psi\in\dom(H[k]_ {X\eta,X\tau,X\lambda,0}),
\end{equation*}
that is, for $\psi=\psi_-\oplus\psi_+\in H^1(\R_-;\C^2)\oplus H^1(\R_+;\C^2)$ the condition 
\begin{equation} \label{eq:TCfull}
(2i\sigma_0-\sigma_1 M_{\eta,\tau,\lambda,\omega})\psi(0_+)=(2i\sigma_0+\sigma_1 M_{\eta,\tau,\lambda,\omega})\psi(0_-)
\end{equation}
holds if and only if the condition
\begin{equation} \label{eq:TCred}
(2i\sigma_0-\sigma_1 M_{X\eta,X\tau,X\lambda,0})\psi(0_+)=(2i\sigma_0+\sigma_1 M_{X\eta,X\tau,X\lambda,0})\bar z\psi(0_-)
\end{equation}
is true.
We may assume that $\omega\neq 0$; otherwise, the problem has an obvious solution. Then the matrices on both sides of \eqref{eq:TCfull} are invertible, because
\begin{equation*}
\det(2i\sigma_0\pm\sigma_1 M_{\eta,\tau,\lambda,\omega})=\omega^2-4-d\pm 4i\omega\neq 0.
\end{equation*}
Furthermore, we will assume that $X$ is such that
\begin{equation} \label{eq:Xassump}
\det(2i\sigma_0\pm \sigma_1 M_{X\eta,X\tau,X\lambda,0})=-4-dX^2\neq 0,
\end{equation}
and so we can also invert the matrices on both sides of \eqref{eq:TCred}. Using $XM_{\eta,\tau,\lambda,0}=M_{X\eta,X\tau,X\lambda,0}$, we get that \eqref{eq:unitEq} is equivalent to
\begin{multline} \label{eq:eqCond}
(2i\sigma_0+\sigma_1 M_{\eta,\tau,\lambda,\omega})(2i\sigma_0+X\sigma_1 M_{\eta,\tau,\lambda,0})^{-1}\\
=\bar{z}(2i\sigma_0-\sigma_1 M_{\eta,\tau,\lambda,\omega})(2i\sigma_0-X\sigma_1 M_{\eta,\tau,\lambda,0})^{-1}.
\end{multline}
With the help of \eqref{eq:Msquared}, we infer that
$$(2i\sigma_0\pm X\sigma_1 M_{\eta,\tau,\lambda,0})(2i\sigma_0\mp X\sigma_1 M_{\eta,\tau,\lambda,0})=-4-dX^2.$$
Hence, we get
$$(2i\sigma_0\pm X\sigma_1 M_{\eta,\tau,\lambda,0})^{-1}=\frac{1}{-4-dX^2}(2i\sigma_0\mp X\sigma_1 M_{\eta,\tau,\lambda,0}).$$
Substituting this into \eqref{eq:eqCond}, using \eqref{eq:Msquared} and $M_{\eta,\tau\lambda,\omega}=M_{\eta,\tau,\lambda,0}+\omega\sigma_1$ we arrive at
\begin{multline} \label{eq:compatibility}
(-4-dX+2i\omega)\sigma_0+(-\omega X+2(1-X)i)\sigma_1 M_{\eta,\tau,\lambda,0}\\
=\bar{z}\big((-4-dX-2i\omega)\sigma_0+(-\omega X-2(1-X)i)\sigma_1 M_{\eta,\tau,\lambda,0}\big).
\end{multline}
Since $\sigma_1 M_{\eta,\tau,\lambda,0}$ is a linear combination of $\{\sigma_i\}_{i=1}^{3}$ and the system $\{\sigma_i\}_{i=0}^{3}$ is linearly independent, this may happen if and only if
\begin{equation} \label{equation_coeff}
\begin{split}
4+dX-2i\omega&=\bar{z}(4+dX+2i\omega)\\
\omega X-2(1-X)i&=\bar{z}(\omega X+2(1-X)i).
\end{split}
\end{equation}
This yields
\begin{equation} \label{eq:zDef}
z=\frac{4+dX+2i\omega}{4+dX-2i\omega} \quad\wedge\quad z=\frac{\omega X+2(1-X)i}{\omega X-2(1-X)i}.
\end{equation}
Note that the denominators in \eqref{eq:zDef} are always non-zero, because $\omega\neq 0$, and $|z|=1$, as required.

It remains to prove that $X\in\R$ exists such that the equalities in \eqref{eq:zDef} hold simultaneously. By comparing the two expressions for $z$ in~\eqref{eq:zDef} one concludes that this is equivalent to finding a real solution of
\begin{equation} \label{eq:quadratic}
dX^2+(4-d+\omega^2)X-4=0,
\end{equation}
see also \cite[Theorem 2.1]{CaLoMaTu_21}. If $d\neq 0$ then  we get
\begin{equation} \label{eq:quadraticSol}
X=\frac{1}{2d}\Big(d-4-\omega^2\pm\sqrt{(d-4-\omega^2)^2+16d}\,\Big)
\end{equation}
with
$$(d-4-\omega^2)^2+16d=(d+4-\omega^2)^2+16\omega^2>0,$$
and so there are always two real solutions of \eqref{eq:quadratic}.
If $d=0$ then \eqref{eq:quadratic} reduces to $(4+\omega^2)X-4=0$ with the real solution
\begin{equation} \label{eq:linSol}
X=\frac{4}{4+\omega^2}.
\end{equation}
Finally, if we started our considerations with $X$ given by either \eqref{eq:quadraticSol} or \eqref{eq:linSol}  then  \eqref{eq:Xassump} would be always fulfilled, because otherwise  $d<0$ and $X=\pm 2/\sqrt{-d}$ which would not be compatible with \eqref{eq:quadraticSol}, as we assume $\omega \neq 0$.

We summarize our findings in the following theorem:
\begin{theorem}
Let $\omega\neq 0$, $z$ be defined by either of the equations in \eqref{eq:zDef} with $X$ given by \eqref{eq:quadraticSol} or \eqref{eq:linSol} for $d\neq 0$ or $d=0$, respectively. Then
$$H[k]_{\eta,\tau,\lambda,\omega}=U_z H[k]_ {X\eta,X\tau,X\lambda,0} U_{\bar z},$$
where the unitary mapping $U_z$ is described in \eqref{eq:Udef}.
\end{theorem}

Consequently, in what follows, we will restrict ourselves to the case $\omega=0$. Moreover, to lighten the notation, we will write 
\begin{equation*}
H[k]_{\eta,\tau,\lambda}:= H[k]_{\eta,\tau,\lambda,0}.
\end{equation*}

If $\omega=0$, $X=-4/d$ with $d\notin\{-4,0\}$, and $z=-1$ then \eqref{equation_coeff} and hence \eqref{eq:compatibility} are clearly valid and still equivalent to \eqref{eq:unitEq}. Therefore,
we get
\begin{proposition} \label{prop:UnitTrafo}
Let $d\notin \{-4,0\}$ and $U_{-1}$ be the unitary mapping given by \eqref{eq:Udef}. Then
\begin{equation*}
H[k]_{\eta,\tau,\lambda}=U_{-1}H[k]_{-\frac{4}{d}\eta,-\frac{4}{d}\tau,-\frac{4}{d}\lambda}U_{-1}.
\end{equation*}
\end{proposition}
\noindent This has been already observed in the two-dimensional setting for compact curves in \cite{CaLoMaTu_21} and, in special cases in \cite[Proposition~4.8~(i)]{BHOP20} and in dimension three in \cite[Theorem~2.3~(d)]{HOP17} and \cite[Proposition~4.2~(i)]{BEHL19_1}.

\section{Spectra of fiber operators} \label{sec:spec_fibers}

In this section we investigate the spectrum of the self-adjoint fiber operators $H[k]_{\eta,\tau,\lambda}= H[k]_{\eta,\tau,\lambda,0}$; cf. \cite{BeDa_94, GeSe_87, GrLe_06, PaRi_14} for related considerations.  From now on we shall exclude the confinement case 
and always assume that $d\neq -4$.
Since, by~\eqref{krein_formula}, the difference of the resolvents of $H[k]_{\eta,\tau,\lambda}$ and $H[k]$ is at most a rank two operator, we get
$$\sigma_{\textup{ess}}(H[k]_{\eta,\tau,\lambda})=\sigma_{\textup{ess}}(H[k])=\big(-\infty,-\sqrt{m^2+k^2}\big]\cup\big[\sqrt{m^2+k^2},+\infty\big)$$ 
due to Weyl's essential spectrum theorem. In fact mimicking the proof of \cite[Proposition 2.3]{PaRi_14} and using~\eqref{krein_formula} one can even show that
\begin{eqnarray*}
&\sigma_{\textup{ac}}(H[k]_{\eta,\tau,\lambda})=\big(-\infty,-\sqrt{m^2+k^2}\big]\cup\big[\sqrt{m^2+k^2},+\infty\big),\\
&\sigma_{\textup{sc}}(H[k]_{\eta,\tau,\lambda})=\emptyset,\,\,\text{ and }\,\,\sigma_{\textup{p}}(H[k]_{\eta,\tau,\lambda})\subset(-\sqrt{m^2+k^2},\sqrt{m^2+k^2}),
\end{eqnarray*}
i.e., outside the gap of $\sigma(H[k])$, the spectrum of $H[k]_{\eta,\tau,\lambda}$ is purely absolutely continuous.
Furthermore, since $\sigma(H[k])\cap (-\sqrt{m^2+k^2},\sqrt{m^2+k^2})=\emptyset$ by \eqref{spechk} and the difference of the resolvents 
of $H[k]_{\eta,\tau,\lambda}$ and $H[k]$ is at most a rank two operator it is also clear that there are at most two simple
eigenvalues or one eigenvalue of multiplicity at most two of $H[k]_{\eta,\tau,\lambda}$ inside
$(-\sqrt{m^2+k^2},\sqrt{m^2+k^2})$; cf. \cite[Chapter~9.3, Theorem~3]{BS}. 
In the rest of this section we will investigate these eigenvalues in more detail (and, actually, it will turn out that eigenvalues of multiplicity two do not exist in the present situation).

Take $z\in(-\sqrt{m^2+k^2},\sqrt{m^2+k^2})$ and look for non-trivial solutions of
\begin{equation*}
 H[k]_{\eta,\tau,\lambda}\psi=z\psi.
\end{equation*}
For $z\neq -m$, the corresponding differential equation has the general solution
\begin{equation*}
 \psi(x)=C\begin{pmatrix} 
	      1\\ i \frac{k+\mu}{z+m}
             \end{pmatrix}
\ena{-\mu x}+D\begin{pmatrix}
              1\\  i\frac{k-\mu}{z+m}
              \end{pmatrix}
\ena{\mu x},
\end{equation*}
where $\mu:=\sqrt{m^2+k^2-z^2}>0$. For $z=-m$ (and hence $k\neq 0$) the solution is
\begin{equation*}
 \psi(x)=C\begin{pmatrix} 
	      1\\ -i\frac{m}{k}
             \end{pmatrix}
\ena{k x}+D\begin{pmatrix}
              0\\  1
              \end{pmatrix}
\ena{-k x}.
\end{equation*}

Now, the integrability condition for $\psi$ yields
\begin{equation*}
 \psi(x)=C\begin{pmatrix} 
	      1\\ i \frac{k+\mu}{z+m}
             \end{pmatrix}
\Theta(x)\ena{-\mu x}+D\begin{pmatrix}
              1\\  i\frac{k-\mu}{z+m}
              \end{pmatrix}
\Theta(-x)\ena{\mu x}
\end{equation*}
or
\begin{equation*}
 \psi(x)=C\begin{pmatrix} 
	      1\\ -i\frac{m}{k}
             \end{pmatrix}
\Theta(-\sgn(k)x)\ena{k x}+D\begin{pmatrix}
              0\\  1
              \end{pmatrix}
\Theta(\sgn(k)x)\ena{-k x},
\end{equation*}
for $z\neq -m$ or $z=-m$, respectively; here we have used 
\begin{equation*}
 \Theta(x)=\begin{cases} 1 & \text{if }x>0, \\ 0 & \text{if }x\leq 0.\end{cases}
\end{equation*}

The constants $C$ and $D$ must be chosen in a way that the transmission condition \eqref{eq:TClambda} at $x=0$ holds true. For $z\neq -m$, \eqref{eq:TClambda} is equivalent to
\begin{equation} \label{eq:ev_gen1}
 C\begin{pmatrix} 
	      1\\ i \frac{k+\mu}{z+m}
             \end{pmatrix}=
            D \Lambda \begin{pmatrix}
              1\\  i\frac{k-\mu}{z+m}
              \end{pmatrix},
\end{equation}
and, for $z=-m$, it is equivalent to
\begin{equation*} 
 C\begin{pmatrix} 
	      1\\ - i\frac{m}{k}
             \end{pmatrix}=
            D \Lambda \begin{pmatrix}
              0\\  1
              \end{pmatrix}
\quad\text{or}\quad  
D\begin{pmatrix} 
	      0\\ 1
             \end{pmatrix}=
            C \Lambda \begin{pmatrix}
              1\\  -i\frac{m}{k}
              \end{pmatrix},
\end{equation*}
if $k<0$ or $k>0$, respectively. Since $\omega=0$, the matrix $\Lambda$ in \eqref{eq:LambdaDef} has the form
\begin{equation*}
\Lambda=\frac{1}{d+4}\begin{pmatrix}
4+4\lambda-d &4i(\tau-\eta)\\
-4i(\tau+\eta)& 4-4\lambda-d
\end{pmatrix}. 
\end{equation*}

The discrete eigenvalues of $H[k]_{\eta,\tau,\lambda}$ different from $-m$ are exactly those values of $z\in(-\sqrt{m^2+k^2},\sqrt{m^2+k^2})$ such that \eqref{eq:ev_gen1} has a non-trivial solution $(C,D)$. 
This happens if and only if
$0 = \det R$, where $R$ is the matrix with columns $\big(\begin{smallmatrix}1\\ i \frac{k+\mu}{z+m} \end{smallmatrix} \big)$ and $- \Lambda \big( \begin{smallmatrix} 1\\  i\frac{k-\mu}{z+m}\end{smallmatrix} \big)$,
which is equivalent to
\begin{equation*}
i\frac{k+\mu}{z+m}\left(4+4\lambda-d+4i(\tau-\eta)i\frac{k-\mu}{z+m}\right)=-4i(\eta+\tau)+(4-4\lambda-d)i\frac{k-\mu}{z+m},
\end{equation*}
and can be rearranged as
\begin{equation} \label{eq:char_eq}
 \sqrt{m^2+k^2-z^2}~(d-4)=4(\eta z+\lambda k+\tau m),
\end{equation}
where $\mu=\sqrt{m^2+k^2-z^2}$ is used.
Similarly as above, if $k<0$ then $z=-m$ is an eigenvalue of $H[k]_{\eta,\tau,\lambda}$ if and only if
$$4-4\lambda-d=4(\tau-\eta)\frac{m}{k},$$
which is just \eqref{eq:char_eq} with $z=-m$ and $k<0$. If $k>0$ then $z=-m$ is an eigenvalue of $H[k]_{\eta,\tau,\lambda}$ if and only if
$$4+4\lambda-d=-4(\tau-\eta)\frac{m}{k}.$$
This equality is equivalent to \eqref{eq:char_eq} with $z=-m$ and $k>0$. Consequently, in all cases it is sufficient to study \eqref{eq:char_eq} when looking for the eigenvalues of $H[k]_{\eta,\tau,\lambda}$. The corresponding (non-normalized) eigenfunctions are given as
\begin{equation} \label{eq:EF1}
\psi(x)= \Lambda \begin{pmatrix}
              1\\  i\frac{k-\mu}{z+m}
              \end{pmatrix}
\Theta(x)\ena{-\mu x}+\begin{pmatrix}
              1\\  i\frac{k-\mu}{z+m}
              \end{pmatrix}
\Theta(-x)\ena{\mu x}
\end{equation}
and
\begin{equation} \label{eq:EF2}
 \psi(x)=\begin{cases}
 \Lambda \begin{pmatrix}
              0\\  1
              \end{pmatrix}
\Theta(x)\ena{k x}+\begin{pmatrix}
              0\\  1
              \end{pmatrix}
\Theta(-x)\ena{-k x} & \text{ if } k<0,\\
\begin{pmatrix} 
	      1\\ -i\frac{m}{k}
             \end{pmatrix}
\Theta(-x)\ena{k x}+\Lambda \begin{pmatrix}
              1\\  -i\frac{m}{k}
              \end{pmatrix}
\Theta(x)\ena{-k x} & \text{ if } k>0,
\end{cases}
\end{equation}
for $z\neq -m$  and $z=-m$, respectively.

\subsection{The case $d=4$} \label{sec:d=4}
Note that 
\begin{equation} \label{eq:d4}
4=d=\eta^2-\tau^2-\lambda^2
\end{equation}
implies $\eta\neq 0$. 
Hence, if $d=4$ then \eqref{eq:char_eq} has exactly one solution 
\begin{equation} \label{eq:LinBand}
z=-\frac{\lambda k+\tau m}{\eta}.
\end{equation}

Moreover, using \eqref{eq:d4} and \eqref{eq:LinBand} one verifies that $z^2<m^2+k^2$ is equivalent to $2\lambda\tau km<(4+\lambda^2)m^2+(4+\tau^2)k^2$, which holds true except for the case $m=k=0$, due to the Young inequality. Consequently, for $m\neq 0$ and $k\in\R$ we observe that $z\in(-\sqrt{m^2+k^2},\sqrt{m^2+k^2})$; for $m=0$ this holds true for all $k\in\R\setminus\{0\}$.

\subsection{The case $d\neq 4$}
For the existence of a solution $z\in(-\sqrt{m^2+k^2},\sqrt{m^2+k^2})$ to \eqref{eq:char_eq} we necessarily need
\begin{equation} \label{eq:EVnecess}
(d-4)(\eta z+\lambda k+\tau m)>0.
\end{equation}
Squaring \eqref{eq:char_eq} we get the following quadratic equation in $z$: 
\begin{equation} \label{eq:char_eqSquared}
\Big(\eta^2+\Big(\frac{d}{4}-1\Big)^2\Big)z^2+2\eta(\lambda k+\tau m)z+(\lambda k+\tau m)^2-(m^2+k^2)\Big(\frac{d}{4}-1\Big)^2=0,
\end{equation}
where the discriminant is 
\begin{equation} \label{eq:diskr}
4 \left( \frac{d}{4}-1 \right)^2 \left[\Big(\eta^2+\Big(\frac{d}{4}-1\Big)^2-\lambda^2\Big)k^2-2\lambda\tau m k+m^2\Big(\eta^2+\Big(\frac{d}{4}-1\Big)^2-\tau^2\Big)\right]. 
\end{equation}
Note that 
\begin{equation*} 
\eta^2+\Big(\frac{d}{4}-1\Big)^2-\lambda^2=\tau^2+\Big(\frac{d}{4}+1\Big)^2>0
\end{equation*}
and
$$\eta^2+\Big(\frac{d}{4}-1\Big)^2-\tau^2=\lambda^2+\Big(\frac{d}{4}+1\Big)^2>0,$$
and hence the expression in the square brackets in \eqref{eq:diskr} can be rewritten in the form 
$$
\Big(\tau^2+\Big(\frac{d}{4}+1\Big)^2\Big)k^2-2\lambda\tau m k+ \Big(\lambda^2+\Big(\frac{d}{4}+1\Big)^2\Big) m^2.
$$

One verifies that the values of this polynomial in the variable $k$ are strictly positive if $m\not=0$ and non-negative if $m=0$ (in which case there is a zero at $k=0$).
Since the case $m=k=0$ is not allowed in this section (recall that we consider $z\in(-\sqrt{m^2+k^2},\sqrt{m^2+k^2})$) we conclude that the discriminant \eqref{eq:diskr} is
strictly positive and hence \eqref{eq:char_eqSquared} has the following pair of real solutions 
\begin{equation} \label{eq:EV}
z_\pm=\frac{-\eta(\lambda k+\tau m)\pm \big|\frac{d}{4}-1\big|\sqrt{\big(\tau^2+\big(\frac{d}{4}+1\big)^2\big)k^2-2\lambda\tau m k+\big(\lambda^2+\big(\frac{d}{4}+1\big)^2\big)m^2}}{\eta^2+\big(\frac{d}{4}-1\big)^2}.
\end{equation}
For all real solutions of \eqref{eq:char_eqSquared} that satisfy \eqref{eq:EVnecess}, the equation \eqref{eq:char_eq} holds and all solutions of \eqref{eq:char_eq} are clearly in $[-\sqrt{m^2+k^2}, \sqrt{m^2+k^2}]$. Moreover, $z=\pm\sqrt{m^2+k^2}$ cannot obey \eqref{eq:EVnecess} and \eqref{eq:char_eqSquared} simultaneously. Therefore, the solutions $z_\pm$ in \eqref{eq:EV} satisfying \eqref{eq:EVnecess} lie in $(-\sqrt{m^2+k^2}, \sqrt{m^2+k^2})$.

\begin{remark} \label{rem:band_prop} 
For $m\neq 0$, one verifies that the functions $z_\pm=z_\pm(k)$ are either strictly convex or strictly concave (by computing the second derivatives and using  that the expression in \eqref{eq:diskr} is strictly positive). Consequently, there are five possibilities for the domains of $z_\pm=z_\pm(k)$ that are given exactly by those $k$'s for which \eqref{eq:EVnecess} holds. Namely, these functions are defined either on a bounded open interval or a union of two disjoint  unbounded open intervals or a semi-bounded open interval or $\R$ or nowhere. 
The case $m=0$ is discussed later in Section \ref{sec:dneq4}, where we will show that $z_\pm=z_\pm(k)$ are defined on  certain unions of sets $\emptyset,\, (-\infty,0),$ and $(0,+\infty)$.
\end{remark}

We summarize the results obtained in this section in the following theorem:
\begin{theorem} \label{theo:spectra_fibers}
Let $m,k\in\R$ and assume that $\eta,\tau,\lambda\in\R$ are such that $d\neq -4$. Then 
\begin{eqnarray*}
&\sigma_{\textup{ac}}(H[k]_{\eta,\tau,\lambda})=\bigl(-\infty,-\sqrt{m^2+k^2}\bigr]\cup\bigl[\sqrt{m^2+k^2},+\infty\bigr)\\
&\sigma_{\textup{sc}}(H[k]_{\eta,\tau,\lambda})=\emptyset,\,\,\text{ and }\,\,\sigma_{\textup{p}}(H[k]_{\eta,\tau,\lambda})\subset(-\sqrt{m^2+k^2},\sqrt{m^2+k^2}).
\end{eqnarray*}
Furthermore, if $m\not=0$ or $k\not=0$, then $(-\sqrt{m^2+k^2},\allowbreak \sqrt{m^2+k^2})$ is a gap in the essential spectrum of $H[k]_{\eta,\tau,\lambda}$ and there are at most two isolated simple eigenvalues
of $H[k]_{\eta,\tau,\lambda}$ inside this gap:
\begin{itemize}
 \item [{\rm (i)}] If $d=4$, then $z$ given in \eqref{eq:LinBand} is the only eigenvalue of $H[k]_{\eta,\tau,\lambda}$.
 \item [{\rm (ii)}] If $d\neq 4$, then exactly those $z_\pm$ given in \eqref{eq:EV} that obey \eqref{eq:EVnecess} are eigenvalues of $H[k]_{\eta,\tau,\lambda}$.
\end{itemize}
In both cases, the corresponding eigenfunctions are given by \eqref{eq:EF1} and \eqref{eq:EF2}.
\end{theorem}

\section{Two-dimensional Dirac operators with a $\delta$-shell interaction}
\label{sec:2dOp}

In this section we investigate the two-dimensional Dirac operator with a $\delta$-shell interaction supported on the straight line $\{(0,y)|\, y\in\R\}$ 
associated with the formal differential expression \eqref{eq:formal}. For $\eta,\tau,\lambda\in\R$ such that $d=\eta^2-\tau^2-\lambda^2\not=-4$, we first define an operator $H_{\eta,\tau,\lambda}$ with the help of the direct integral of the fiber operators
$H[k]_{\eta,\tau,\lambda}$ by
\begin{equation*}
\begin{split}
 (H_{\eta,\tau,\lambda}\psi)(k)&=H[k]_{\eta,\tau,\lambda}\psi(\cdot,k),\\
\dom(H_{\eta,\tau,\lambda})&=\Big\{\psi\in L^2(\R^2,\dd x\dd k;\C^2)|\, \psi(\cdot,k)\in\dom(H[k]_{\eta,\tau,\lambda})\, \text{a.e.},\\
&\qquad\qquad\qquad\qquad\qquad\qquad\int_{\R}\|H[k]_{\eta,\tau,\lambda}\psi(\cdot,k)\|^2\dd k<\infty\Big\},
\end{split} 
\end{equation*}
where $L^2(\R^2,\dd x\dd k;\C^2)=\int_\R^\oplus L^2(\R,\dd x;\C^2)\dd k$;
cf. Appendix~\ref{section_appendix} for a brief summary of the direct integral of Hilbert spaces and self-adjoint operators. Since the fiber operators $H[k]_{\eta,\tau,\lambda}$ are self-adjoint in $L^2(\R;\C^2)$ and  $k\mapsto H[k]_{\eta,\tau,\lambda}$ is measurable (in the sense that $k\mapsto \langle f,(H[k]_{\eta,\tau,\lambda}-i)^{-1}g\rangle_{L^2(\mathbb{R}; \mathbb{C}^2)}$ is measurable), the operator 
$H_{\eta,\tau,\lambda}$ is self-adjoint in $L^2(\R^2;\C^2)$. 
Now the Dirac operator with a $\delta$-potential supported on the straight line is given by
\begin{equation}\label{diracii}
\hat{H}_{\eta,\tau,\lambda} := \mathcal{F}_{y \rightarrow k}^{-1} H_{\eta,\tau,\lambda} \mathcal{F}_{y \rightarrow k},
\end{equation}
and since  $\mathcal{F}_{y \rightarrow k}$ is unitary it suffices to study the spectral properties of $H_{\eta,\tau,\lambda}$.
Furthermore, in view of~\eqref{spectrum_direct_integral} and~\eqref{point_spectrum} the spectral analysis of $H_{\eta,\tau,\lambda}$ 
reduces to the spectral analysis of the one-parametric family of one-dimensional operators $H[k]_{\eta,\tau,\lambda}$. More precisely, in the present situation we have
$z\in \sigma(H_{\eta,\tau,\lambda})$ if and only 
\begin{equation*}
 \bigl|\bigl\{ k \in \R\,|\, \sigma(H[k]_{\eta,\tau,\lambda}) \cap (z-\varepsilon, z+\varepsilon) \neq \emptyset\bigr\}\bigr| > 0 \,\,\text{for all}\,\,\varepsilon>0
\end{equation*}
and $z\in \sigma_\textup{p}(H_{\eta,\tau,\lambda})$ if and only 
\begin{equation}\label{use2}
\bigl|\bigl\{ k \in \R\,|\, z \in \sigma_\textup{p}(H[k]_{\eta,\tau,\lambda})\bigr\}\bigr| > 0;
\end{equation}
here $|B|$ denotes the Lebesgue measure of $B \subset \mathbb{R}$.
Recall also from the discussion after \eqref{point_spectrum} that each eigenvalue of $H_{\eta,\tau,\lambda}$ has infinite multiplicity.
Combining this with Theorem \ref{theo:spectra_fibers} and Corollary  \ref{corollary_ac_applicable} we will get a full picture of the spectrum of $H_{\eta,\tau,\lambda}$ in 
Theorem~\ref{theo:full_spectrum} below. 

\begin{remark}
In the purely electrostatic case, i.e., when $\tau=\lambda=0$, the spectrum of the 
operator $\mathcal{F}_{y \rightarrow k}^{-1} H_{\eta,0,0} \mathcal{F}_{y \rightarrow k}$ was studied in \cite{BeHoTu_21} 
with the help of boundary triples and Weyl functions. The statements proved there are recovered in the present paper; cf. Section~\ref{section_electrostatic}.
We also note that the methods used in \cite{BeHoTu_21} allow a more explicit description of the domain of the Dirac operator \eqref{diracii} in terms of traces.   
\end{remark}

\subsection{Eigenvalues}
\subsubsection{The case $d=4$} If $d=4$ then there is a single band $z=z(k)$ in $\sigma_\textup{p}(H[k]_{\eta,\tau,\lambda})$ given by \eqref{eq:LinBand}. For $\lambda\neq 0$, it is linear and non-constant. Therefore, it does not contribute to the point spectrum of $H_{\eta,\tau,\lambda}$. If $\lambda=0$ then the band  is constant  and defined on $\R$ or $\R\setminus\{0\}$ for $m\neq 0$ or $m=0$, respectively. In both cases, we conclude from \eqref{use2} that
$$z=-\frac{\tau m}{\eta}$$
is an eigenvalue of $H_{\eta,\tau,\lambda}$ of infinite multiplicity. 

\subsubsection{The case $d\neq 4$} \label{sec:dneq4} If $d\neq 4$ (recall that the case $d\neq -4$ is not investigated) then according to Theorem \ref{theo:spectra_fibers}, there are at most two energy bands $z_\pm$ given by \eqref{eq:EV} that may contribute to an eigenvalue of the full operator. 
If $m\neq 0$ then the function $z_\pm=z_\pm(k)$ is always non-linear and the equation $z_\pm(k)=C$ has at most two solutions for any constant $C\in\R$. Therefore, there cannot be any eigenvalue in the spectrum of the full operator in this case. 

If $m=0$ then \eqref{eq:EV} simplifies to
\begin{equation} \label{eq:LinBand2}
z_\pm =\frac{-\eta\lambda k\pm|\frac{d}{4}-1|\sqrt{\tau^2+(\frac{d}{4}+1)^2}~|k|}{\eta^2+(\frac{d}{4}-1)^2}.
\end{equation}
Only $z_\pm\in(-|k|,|k|)$ that satisfy \eqref{eq:EVnecess}, i.e., 
\begin{equation} \label{condition_m_0} 
  (d-4)\left(\eta \frac{-\eta\lambda \pm \sgn k |\frac{d}{4}-1|\sqrt{\tau^2+(\frac{d}{4}+1)^2}}{\eta^2+(\frac{d}{4}-1)^2}+\lambda \right) k>0,
\end{equation}
are eigenvalues of $H[k]_{\eta,\tau,\lambda}$. Clearly, if~\eqref{condition_m_0}  is satisfied for one $k \in (0, +\infty)$, then~\eqref{condition_m_0}  holds for all $k \in (0, +\infty)$; a similar consideration holds for $k \in (-\infty, 0)$.
We see now that  $z_\pm=z_\pm(k)$ is a linear function on $(-\infty,0)$ or $(0,+\infty)$, respectively. Therefore, if $z_\pm$ is an eigenvalue of $H_{\eta,\tau,\lambda}$ then $z_\pm=z_\pm(k)$ must be constant on $(-\infty,0)$ or $(0,+\infty)$, respectively. Squaring the necessary condition \eqref{eq:char_eq} for the discrete eigenvalues of $H[k]_{\eta,\tau,\lambda}$ we get
$$(k^2-z_\pm^2)(d-4)^2=16(\lambda^2k^2+2\lambda\eta z_\pm k+\eta^2 z_\pm^2).$$
Assuming that $z_\pm=z_\pm(k)$ is constant this yields
\begin{equation} \label{eq:quadr_coeff}
(d-4)^2=16\lambda^2,\quad \lambda\eta z_\pm=0,\quad -(d-4)^2 z_\pm^2=16\eta^2 z_\pm^2.
\end{equation}
From the last equation it follows that either $z_\pm=0$  or  $-(d-4)^2=16\eta^2$. Taking the first equation in \eqref{eq:quadr_coeff} into account, the latter equality implies $\lambda=\eta=0$ and $d=4$, which is not possible. Hence, it remains to consider the case  $z_\pm\equiv z=0\,(=m)$. Then \eqref{eq:char_eq} and \eqref{eq:EVnecess} take the form
\begin{equation} \label{eq:embedded_cond}
(d-4)|k|=4\lambda k \quad \text{and} \quad (d-4)\lambda k>0,
\end{equation}
respectively. If $d-4=4\lambda$ then \eqref{eq:embedded_cond} holds true for all $k>0$. If $4-d=4\lambda$ then \eqref{eq:embedded_cond} holds true for all $k<0$. We conclude that $z=0$ is an  eigenvalue of infinite multiplicity of $H_{\eta,\tau,\lambda}$, whenever $m=0$ and $0\neq d-4=\pm 4\lambda$. 

\subsection{Continuous (bulk) spectrum} The band $z=z(k)$ in~\eqref{eq:LinBand} is everywhere defined, except for the case $m=0$ when it is defined on $\R\setminus\{0\}$, and the functions $z_\pm=z_\pm(k)$ given by~\eqref{eq:EV} are defined on unions of at most two disjoint open intervals, see Remark \ref{rem:band_prop}. Moreover, $z$ and $z_\pm$ are real-analytic on their domains of definition and hence, they
obey the assumptions of Corollary \ref{corollary_ac_applicable}. Thus, we infer that $\sigma_{\textup{sc}}(H_{\eta,\tau,\lambda})=\emptyset$. Moreover, for $d=4$,
$$\sigma(H_{\eta,\tau,\lambda})=(-\infty,-|m|]\cup[|m|,+\infty)\cup \overline{\ran(z)}, $$
and for $d\neq 4$,
$$\sigma(H_{\eta,\tau,\lambda})=(-\infty,-|m|]\cup[|m|,+\infty)\cup \overline{\ran(z_+)}\cup \overline{\ran(z_-)}.$$

If $\sigma_{\textup{p}}(H_{\eta,\tau,\lambda})=\emptyset$ then $\sigma(H_{\eta,\tau,\lambda})=\sigma_{\textup{ac}}(H_{\eta,\tau,\lambda})$. We have seen in the  previous subsection that in special cases there is a single eigenvalue due to a constant band. It may be either embedded in the absolutely continuous spectrum or isolated. In the first case we still have $\sigma(H_{\eta,\tau,\lambda})=\sigma_{\textup{ac}}(H_{\eta,\tau,\lambda})$, because the absolutely continuous spectrum is closed by definition. The full description of  $\sigma(H_{\eta,\tau,\lambda})$ is  as follows:

\begin{theorem} \label{theo:full_spectrum}
Let $m\in\R$ and assume that $\eta,\tau,\lambda\in\R$ are such that $d\neq -4$. Then 
$\sigma_{\textup{sc}}(H_{\eta,\tau,\lambda})=\emptyset$ and the following holds:
\begin{itemize}
\item [{\rm (i)}] If $d=4$ and $\lambda\neq 0$, then $\sigma_{\textup{ac}}(H_{\eta,\tau,\lambda})=\R$ and  $\sigma_{\textup{p}}(H_{\eta,\tau,\lambda})=\emptyset$.  
\item [{\rm (ii)}] If $d=4$ and $\lambda=0$, then $\sigma_{\textup{ac}}(H_{\eta,\tau,\lambda})=(-\infty,-|m|]\cup[|m|,+\infty)$. Moreover, $z=-\tau m/\eta$ is an eigenvalue of infinite multiplicity, which is isolated in $(-|m|,|m|)$ for $m\neq 0$ and embedded for $m=0$. 
\item [{\rm (iii)}] If $d\neq 4$, then $\sigma_{\textup{ac}}(H_{\eta,\tau,\lambda})=(-\infty,-|m|]\cup[|m|,+\infty)\cup \overline{\ran(z_+)}\cup \overline{\ran(z_-)}.$ Moreover, $\sigma_{\textup{p}}(H_{\eta,\tau,\lambda})=\emptyset$, except when $m=0$ and $d-4=\pm 4\lambda\, (\neq 0)$ in which case $z=0$ is an  embedded eigenvalue of infinite multiplicity.
\end{itemize}
\end{theorem}
We refer to \cite{GrLe_06} for a similar result in the (decoupled) case $d=-4$.

\subsection{Linear bands} \label{sec:LinBands}
Quantum mechanical interpretation of the eigenvalues of a Hamiltonian is quite straightforward--these are just energies of the stationary states (which are described by the respective eigenfunctions). The points of the absolutely continuous spectrum are energies at which the system described by the Hamiltonian exhibits transport, see \cite{BrJaLaPi_16} for a possible mathematical explanation of this relationship. In fact, if the Hamiltonian is invariant with respect to translations in one direction, one can apply the partial Fourier transform to decompose the operator into a direct integral and then investigate some finer  properties of such a transport. In particular, there is a  direct way how to construct normalizable low-dispersing wave packets, which are also referred to as the edge states, that propagate in the direction of the symmetry using the eigenfunctions associated with an energy band \cite{JaTu_17}. For linear bands, such wave packets do not disperse at all! Their group velocity is given by $v=\dd z/\dd k$, where $z=z(k)$ stands for the energy band and $k$ is the momentum in the direction of the symmetry.

However, only few systems with exactly linear bands are known so far \cite{GrLe_06,JaTu_17,PeNe_09}. Therefore, it is remarkable that the  operator $H_{\eta,\tau,\lambda}$ considered in this paper may possess linear bands for specific choices of the coupling constants.
Firstly, we always get the linear band \eqref{eq:LinBand} when $d=4$. The group velocity of the dispersion-less wave packets equals 
$$v=\frac{\dd z}{\dd k}=-\frac{\lambda}{\eta}\in(-1,1);$$
introducing the physical units to our model, this would mean that $|v|<c$ (the speed of light) or $|v|<v_F$ (the Fermi velocity) if we describe a relativistic particle or electronic states in a Dirac material, respectively. Secondly, if $d\neq 4$ and $m=0$, there may be linear bands \eqref{eq:LinBand2} defined on $(-\infty,0)$ or on $(0,+\infty)$; cf.~\eqref{condition_m_0} and the discussion following this equation.

\subsection{Special cases}
Looking more closely at several important examples, we will show in this section that tuning the coupling constants one can change the spectrum of the free operator $\hat{H}$ dramatically. Applying Theorem \ref{theo:full_spectrum} with the unitarily equivalent operator $H_{0,0,0}$ we infer that
$$\sigma(\hat{H})=\sigma_{\textup{ac}}(\hat{H})=(-\infty,-|m|]\cup[|m|,+\infty),\quad \sigma_{\textup{sc}}(\hat{H})=\sigma_{\textup{p}}(\hat{H})=\emptyset.$$
We will see that the gap $(-|m|,|m|)$ may be shrunk arbitrarily from the top, from  the bottom or even from both endpoints simultaneously. In particular, it is possible to close the gap. On the other hand, it is also possible to create an eigenvalue of infinite multiplicity anywhere in the gap, see the second point of Theorem \ref{theo:full_spectrum}. 

\subsubsection{Purely electrostatic interaction} \label{section_electrostatic}
Let $\tau=\lambda=0$. Then $d=4$ if and only if $\eta=\pm 2$. Substituting this into \eqref{eq:LinBand} we get $z=z(k)\equiv 0$. If $\eta\neq \pm 2$ then the solutions \eqref{eq:EV} that obey \eqref{eq:EVnecess} constitute the energy bands. One can rewrite \eqref{eq:EVnecess} as follows
\begin{equation*}
\sgn z=\begin{cases}
\sgn\eta & \text{ for }\eta^2>4\\
-\sgn\eta & \text{ for }\eta^2<4.
\end{cases}
\end{equation*}
Furthermore, \eqref{eq:EV} simplifies to
\begin{equation*}
z_\pm=\pm\frac{|4-\eta^2|}{4+\eta^2}\sqrt{m^2+k^2}.
\end{equation*}
Therefore, there is only one eigenvalue of $H[k]_{\eta,0,0}$ in the gap of $\sigma_\textup{ess}(H[k]_{\eta, 0,0})$, namely
\begin{equation*}
z=\sgn\eta\, \frac{\eta^2-4}{\eta^2+4}\sqrt{m^2+k^2}.
\end{equation*}
Consequently, we have
\begin{equation} \label{eq:spec_es}
\sigma(H_{\eta,0,0})=\begin{cases}
(-\infty,-|m|]\cup\{0\}\cup[|m|,+\infty) & \text{ for }\eta=\pm 2\\
(-\infty,\frac{4-\eta^2}{4+\eta^2}|m|]\cup[|m|,+\infty) & \text{ for }\eta\in(-\infty,-2)\\
(-\infty,-|m|]\cup[\frac{4-\eta^2}{4+\eta^2}|m|,+\infty) & \text{ for }\eta\in(-2,0)\\
(-\infty,\frac{\eta^2-4}{\eta^2+4}|m|]\cup[|m|,+\infty) & \text{ for }\eta\in(0,2)\\
(-\infty,-|m|]\cup[\frac{\eta^2-4}{\eta^2+4}|m|,+\infty) & \text{ for }\eta\in(2,+\infty).\\
\end{cases}
\end{equation}
This has been already observed in \cite{BeHoTu_21} employing a different approach.
See Figure~\ref{fig:es} for plots of the energy bands in several typical situations.
Note that at $\eta=\pm 2$ the spectrum changes quite dramatically.

\begin{figure}[h] 
\begin{center}
   \includegraphics[width=\textwidth]{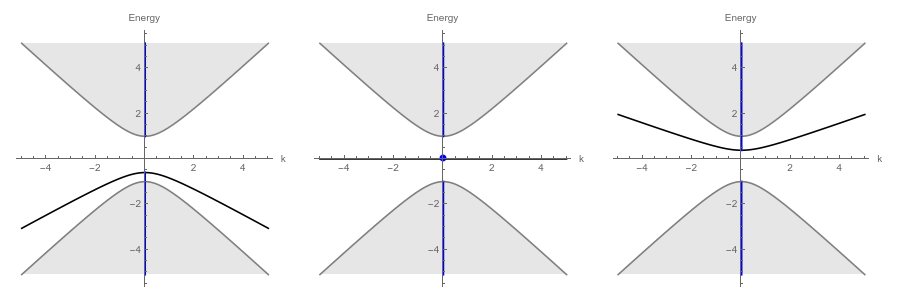}  
   \caption{Purely electrostatic case with $m=1$ and $\eta=1$ (left), $\eta=2$ (middle), and  $\eta=3$ (right). The bulk spectrum bounded by the curves $\pm\sqrt{m^2+k^2}$ is depicted in gray. The thick black line shows the energy band. The spectrum of the full operator $H_{\eta,0,0}$ is presented in blue. Note that for $\eta=2$ the energy band is identically zero which implies that zero is an eigenvalue of infinite multiplicity in the spectrum of $H_{\eta,0,0}$. The same holds true when $\eta=-2$.} \label{fig:es}
\end{center} 
\end{figure}

\subsubsection{Purely Lorentz scalar interaction} \label{sec:magnetic}
Let $\eta=\lambda=0$. Then $d=-\tau^2\neq 4$. If $\tau m<0$ then there are  eigenvalues 
\begin{equation*}
z_\pm=\pm\sqrt{\Big(\frac{4-\tau^2}{4+\tau^2}\Big)^2m^2+k^2}
\end{equation*}
of $H[k]_{0,\tau,0}$ in the gap of $\sigma_\textup{ess}(H[k]_{0,\tau,0})$. If  $\tau m\geq 0$ then there are no eigenvalues in this gap. Note that these results are true also in the decoupled case when $\tau=\pm 2$ \cite{GrLe_06}. We conclude that
\begin{equation*}
\sigma(H_{0,\tau,0})=\begin{cases}
(-\infty,-|m|]\cup[|m|,+\infty) & \text{ if } \tau m\geq 0\\
(-\infty,-\frac{|4-\tau^2|}{4+\tau^2}|m|]\cup[\frac{|4-\tau^2|}{4+\tau^2}|m|,+\infty) & \text{ if } \tau m<0.
\end{cases}
\end{equation*}
Figure \ref{fig:ls} depicts the energy bands in several typical situations.

\begin{figure}[h]
\begin{center}
   \includegraphics[width=\textwidth]{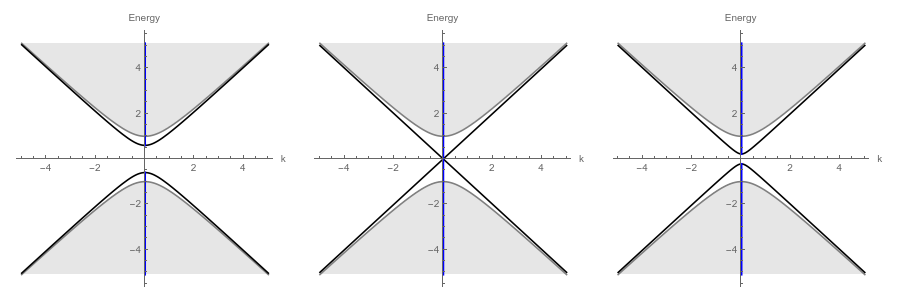}  
   \caption{Purely Lorentz scalar case with $m=1$ and $\tau=-1$ (left), $\tau=-2$ (middle), and  $\tau=-2.5$ (right). The bulk spectrum is depicted in gray and the thick black lines show the energy bands again. If $m>0$ then there are energy bands only for $\tau<0$. The spectrum of the full operator $H_{0,\tau,0}$ is presented in blue. For $\tau=-2$ there are two linear energy bands and the spectrum of $H_{0,\tau,0}$ is the whole real line. The same holds true when $m<0$ and $\tau=2$.} \label{fig:ls}
\end{center}
\end{figure}

\subsubsection{Purely magnetic interaction}
Let $\eta=\tau=0$. Then $d=-\lambda^2\neq 4$. If
\begin{equation*} 
\lambda k<0
\end{equation*}
then 
\begin{equation*}
z_\pm=\pm\sqrt{m^2+\Big(\frac{4-\lambda^2}{4+\lambda^2}\Big)^2 k^2}
\end{equation*}
are eigenvalues of $H[k]_{0,0,\lambda}$ in the gap of $\sigma_\textup{ess}(H[k]_{0,0,\lambda})$. If  $\lambda k\geq 0$ then there are no eigenvalues. Again, these results remain valid also in the decoupled case $\lambda=\pm 2$, as the case $\eta=\tau=0$ and $\lambda=\pm 2$ corresponds to $z=\pm 1$ in \cite{GrLe_06}; then, one can read off the result from equations (22) and (23) in \cite{GrLe_06}. In all cases we have $\sigma(H_{0,0,\lambda})=\sigma(\hat H)=(-\infty,-|m|]\cup[|m|,+\infty)$.  Therefore, it is not possible to reduce the gap in the spectrum of $\hat H$ by means of the magnetic $\delta$-interaction. For $m=0$ and $\lambda k<0$, there are linear bands
$$z_\pm=\pm\frac{|4-\lambda^2|}{4+\lambda^2} |k|.$$
The energy bands for various values of $\lambda$ and $m$ are shown in Figure \ref{fig:magnetic}. Since $\sgn{\lambda}$ determines the orientation of the magnetic field and 
$$k\frac{\dd z_+}{\dd k}>0,\quad  k\frac{\dd z_-}{\dd k}<0,$$
the condition $\lambda k<0$ may be physically interpreted as follows. The "positronic and electronic" low dispersing wave-packets, i.e., packets constructed using the eigenfunctions associated with $z_+$ and $z_-$, respectively, may propagate along the axis $x=0$ in one direction only. This direction is determined by the Lorentz force.

\begin{figure}[h]
\begin{center}
   \includegraphics[width=\textwidth]{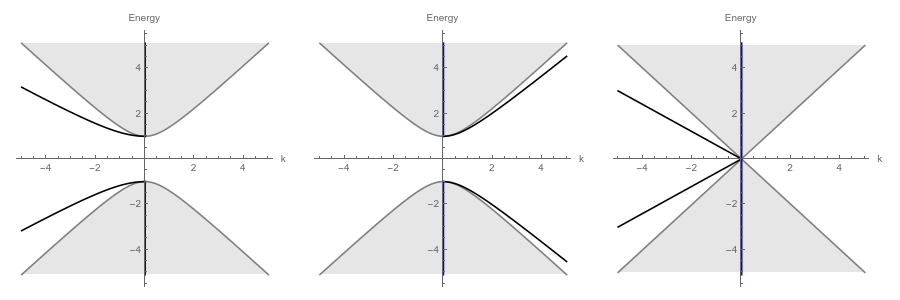}  
   \caption{Purely magnetic  case with  $m=1$ and $\lambda=1$ (left), $m=1$ and $\lambda=-0.5$ (middle), $m=0$ and  $\lambda=1$ (right). The bulk spectrum, the energy bands, and the spectrum of the full operator $H_{0,0,\lambda}$ are depicted in gray, black, and blue, respectively. Note that the bands are supported on an open semi-axis only: for $\lambda>0$ it is the negative semi-axis and vice versa. If $m=0$ there is a linear band. Note that the spectrum of $H_{0,0,\lambda}$ always coincides with the spectrum of the free operator.}
\label{fig:magnetic}  
\end{center}
\end{figure}

\subsubsection{Spectral transition for $d=4$}

As it was noted in Section~\ref{section_electrostatic} there are some cases where a continuous change of the coupling parameters may cause a dramatic change of the spectral properties -- a so-called spectral transition; cf. \cite{BeHoTu_21}. Here we want to discuss another remarkable configuration causing a spectral transition, namely, when $d=4$. By Theorem~\ref{theo:full_spectrum}, if $d=4$ and $\lambda \neq 0$, then one always has
\begin{equation*}
  \sigma(H_{\eta, \tau, \lambda}) = \sigma_{\textup{ac}}(H_{\eta, \tau, \lambda}) = \mathbb{R} \quad \text{and} \quad 
  \sigma_{\textup{p}}(H_{\eta, \tau, \lambda}) = \sigma_{\textup{sc}}(H_{\eta, \tau, \lambda}) = \emptyset.
\end{equation*}
If $\lambda$ changes its value to zero, while keeping $d=4$ fixed, then
\begin{equation*}
  \sigma_{\textup{ac}}(H_{\eta, \tau, \lambda}) = (-\infty, -|m|] \cup [|m|, +\infty), \,\, \sigma_{\textup{p}}(H_{\eta, \tau, \lambda}) = -\frac{\tau}{\eta}m, \,\, \text{and} \,\, 
  \sigma_{\textup{sc}}(H_{\eta, \tau, \lambda}) = \emptyset.
\end{equation*}
We see that for $\lambda=0$ an eigenvalue of infinite multiplicity appears. This change is even more dramatic for $m \neq 0$, as then the absolutely continuous spectrum in $(-|m|,|m|)$ collapses to the single point $-\frac{\tau}{\eta}m$.

\section{Approximations by regular potentials} \label{sec:approx}

The problem of finding regular approximations when $\omega\neq 0$ may be reduced to the case when $\omega=0$ in a similar vein as described in \cite[Section 8]{CaLoMaTu_21}. Therefore, we still consider only the case $\omega=0$ in this section.

\subsection{Approximations of the fiber operators} \label{sec:1Dapprox}
First, one can try to approximate $H[k]_{\eta,\tau,\lambda}$ by the one-dimensional Dirac operator with a scaled regular potential which converges to the $\delta$-distribution in the sense of distributions.  This has been done before when $k=0$. The first rigorous results are due to \v{S}eba \cite{Se_89}, where approximations in the norm resolvent sense were provided for the purely electrostatic and the purely Lorentz scalar $\delta$-interaction. Approximations of a general $\delta$-interaction were found in \cite{Hu_97} and \cite{Hu_99}; however, only in the strong resolvent sense. The norm resolvent convergence of these approximations was proved in \cite{Tu_20} for a three-parametric family of interactions. In fact, one can show norm resolvent convergence of the approximations  for the whole four-parametric family \cite{Ru_21}. Since norm resolvent convergence is stable with respect to adding a constant hermitian perturbation, 
these results extend to the case $k\neq 0$. 

We will now describe the approximating operators. Let $h\in L^1(\R;\R)\cap L^\infty(\R;\R)$ be such that $\int_\R h(x)\dd x=1$ and for $\varepsilon>0$ define
\begin{equation*}
h_\varepsilon(x):=\frac{1}{\varepsilon}h\Big(\frac{x}{\varepsilon}\Big).
\end{equation*}
In the following, we assume $d>-4$; the way how we treat the case $d<-4$ is discussed after Theorem~\ref{theo:approx1}.
Since $\Lambda$ in the transmission condition \eqref{eq:TClambda} is invertible,  there exists a hermitian matrix $A$ such that 
\begin{equation} \label{exponential_matrix}
\exp(-i\sigma_1 A)=\Lambda.
\end{equation} 
To see this we make a similar calculation as in \cite{CaLoMaTu_21} and \cite{Tu_20},  rewrite
$$\Lambda=\frac{4}{4+d}\Big(\frac{4-d}{4}\sigma_0+\lambda\sigma_3-i\eta\sigma_1-\tau\sigma_2\Big),$$
and use the ansatz
$$A\equiv A_{\tilde\eta,\tilde\tau,\tilde\lambda}=\tilde\eta\sigma_0+\tilde\tau\sigma_3+\tilde\lambda\sigma_2
$$
with
\begin{equation} \label{eq:renorm}
(\tilde\eta,\tilde\tau,\tilde\lambda)=\begin{cases}
\frac{2}{\sqrt{d}}\big(\arctan(\sqrt{d}/2)+l\pi\big)(\eta,\tau,\lambda)\, \text{with }l\in\Z & \text{ for }d>0\\
(\eta,\tau,\lambda) & \text{ for }d=0\\
\frac{2}{\sqrt{-d}}\arctanh(\sqrt{-d}/2)(\eta,\tau,\lambda) & \text{ for }d\in(-4,0).
\end{cases}
\end{equation}
Plugging this into 
\begin{equation*}
  \exp(B) = \exp\Big( \frac{\tr B}{2} \Big) \Big( \cos \nu \sigma_0 + \frac{\sin \nu}{\nu} \Big( B -  \frac{\tr B}{2} \sigma_0 \Big) \Big), \quad \nu := \sqrt{\det B - \Big( \frac{\tr B}{2} \Big)^2},
\end{equation*}
which holds for all 2$\times$2 matrices $B$, we find that~\eqref{exponential_matrix} is fulfilled.
Finally, if we  put
\begin{equation*}
 \begin{split}
  H[k]^\varepsilon_{\tilde\eta,\tilde\tau,\tilde\lambda} &=H[k]+A_{\tilde\eta,\tilde\tau,\tilde\lambda} h_\varepsilon  =\sigma_1(-i\partial_x)+\sigma_2 k+\sigma_3 m+A_{\tilde\eta,\tilde\tau,\tilde\lambda} h_\varepsilon ,\\
  \dom(H[k]^\varepsilon_{\tilde\eta,\tilde\tau,\tilde\lambda})&=\dom(H[k])=H^1(\R;\C^2),
 \end{split}
\end{equation*}
then the following result follows from \cite{Ru_21}, as the bounded perturbation $\sigma_2 k$ does not influence the convergence.

\begin{theorem} \label{theo:approx1}
Let $\eta,\tau,\lambda\in\R$ be such that $d=\eta^2-\tau^2-\lambda^2>-4$
and let $\tilde\eta,\tilde\tau,\tilde\lambda\in\R$ be as in \eqref{eq:renorm}.
Then
\begin{equation*}
\lim_{\varepsilon\to 0}\|(H[k]_{\eta,\tau,\lambda}-z)^{-1}-(H[k]^\varepsilon_{\tilde\eta,\tilde\tau,\tilde\lambda}-z)^{-1}\|=0,\quad z\in\C\setminus\R.
\end{equation*}
\end{theorem}

Note that it is much more difficult to deal with the decoupled case when $d=-4$; cf. \cite{Ru_21}. If $d\in(-\infty,-4)$ one can still use the theorem, after employing the unitary transformation from Proposition \ref{prop:UnitTrafo}, see the discussion below the proof of Theorem~2.5 in \cite{CaLoMaTu_21} for details. In fact, the restriction $d>-4$ is the price we pay for our ansatz for $A$ (and consequently, for the concise formula  \eqref{eq:renorm}).

For the distributional limit of the approximating potential we get
$$\lim_{\varepsilon\to 0}A_{\tilde\eta,\tilde\tau,\tilde\lambda}h_\varepsilon=A_{\tilde\eta,\tilde\tau,\tilde\lambda}\delta=(\tilde\eta\sigma_0+\tilde\tau\sigma_3+\tilde\lambda\sigma_2)\delta.$$ 
Except for the case when $d=0$ this is different from the potential in the formal expression \eqref{eq:formalOp} for $H[k]_{\eta,\tau,\lambda}$. Therefore, the renormalization of the coupling constant occurs during the approximating procedure. 
In the one-dimensional relativistic setting, this effect was described for the first time in \cite{Se_89} and generalized and studied in more detail in \cite{Hu_97,Hu_99,Tu_20}. Remarkably, the coupling constants have to be renormalized in the exactly same manner in the two-dimensional \cite{CaLoMaTu_21} and the three-dimensional \cite{MaPi_17,MaPi_18} setting.

\subsection{Approximations of the full operator}  \label{sec:approx_full}
Let us define
\begin{equation*}
H^\varepsilon_{\tilde\eta,\tilde\tau,\tilde\lambda}:=\int_\R^\oplus H[k]^\varepsilon_{\tilde\eta,\tilde\tau,\tilde\lambda}\dd k.
\end{equation*}
By \cite[Theorem XIII.85]{RS4} we have for $z\in\C\setminus\R$
\begin{align*} 
&(H^\varepsilon_{\tilde\eta,\tilde\tau,\tilde\lambda}-z)^{-1}=\int_\R^\oplus (H[k]^\varepsilon_{\tilde\eta,\tilde\tau,\tilde\lambda}-z)^{-1}\dd k,\\
&(H_{\eta,\tau,\lambda}-z)^{-1}=\int_\R^\oplus (H[k]_{\eta,\tau,\lambda}-z)^{-1}\dd k.
\end{align*}
We will prove the following result:
\begin{proposition} \label{proposition_approximation}
Let $\eta,\tau,\lambda\in\R$ be such that $d=\eta^2-\tau^2-\lambda^2>-4$. Then, for any $z\in\C\setminus\R$ and all $\psi\in \int_\R^\oplus L^2(\R,\dd x;\C^2)\dd k\equiv L^2(\R^2,\dd x\dd k;\C^2)$,
\begin{equation*}
\lim_{\varepsilon\to 0}\|(H_{\eta,\tau,\lambda}-z)^{-1}\psi-(H^\varepsilon_{\tilde\eta,\tilde\tau,\tilde\lambda}-z)^{-1}\psi\|=0,
\end{equation*}
where the coefficients $\tilde\eta,\tilde\tau,$ and $\tilde\lambda$ are given in \eqref{eq:renorm}.
\end{proposition}  

We would like to point out that Proposition~\ref{proposition_approximation} also holds true for the critical interaction strengths, i.e., when $(\frac{d}{4}-1)^2-\lambda^2=0$. It is the first time, that such an approximation result is shown in the critical case in a higher dimensional setting.

\begin{remark}[Purely magnetic interaction] If $\eta=\tau=0$ and $\lambda\in(-2,0)\cup(0,2)$ then the matrix part of the approximating potential is just a multiple of $\sigma_2$,
$$A=\tilde{\lambda}\sigma_2=2\arctanh\frac{\lambda}{2}\sigma_2.$$
(If $|\lambda|>2$ we would get $A=2\arctanh\frac{2}{\lambda}\sigma_2+\pi\sigma_1$.) Therefore, we have
$$\mathscr{F}_{y\to k}^{-1}H_{0,0,\tilde{\lambda}}^\varepsilon\mathscr{F}_{y\to k}=\sigma_1(-i\partial_x)+\sigma_2(-i\partial_y+\tilde\lambda h_\varepsilon)+\sigma_3 m.$$
This is just the two-dimensional Dirac Hamiltonian with a magnetic field supported on the $\varepsilon$-tubular neighbourhood of the line $x=0$. When $m=0$ such operators were used to describe the electron states (near one Dirac point) in graphene under the influence of either a perpendicular magnetic field or a strain \cite{FoGuKa_08,MaVaPe_09,PeNe_09}. Note that the function $h_\varepsilon$ determines the profile of the "magnetic barrier". If $h_\varepsilon=\varepsilon^{-1}\chi_{(-\varepsilon,\varepsilon)}$ then we get a model which is analytically solvable--see \cite[Section 2.1]{MaVaPe_09}. The narrow limit, $\varepsilon\to 0$, was treated only formally in \cite{PeNe_09}, and so the renormalization of the coupling constant was not derived there. Nevertheless, unless $\lambda=0$, it is always necessary.  On the other hand, linear bands for the what we call here purely magnetic $\delta$-interaction were already observed in \cite{PeNe_09}.
\end{remark}

\begin{proof}[Proof of Proposition~\ref{proposition_approximation}]
  First, we note that for any two closed operators $A, B$ acting in a Hilbert space and $z_1,z_2 \in \rho(A) \cap \rho(B)$ the relation
  \begin{equation*}
   \begin{split}
    (&A-z_2)^{-1} - (B-z_2)^{-1} \\
    &= \big( 1+(z_2-z_1) (B-z_2)^{-1} \big) \big[ (A-z_1)^{-1}-(B-z_1)^{-1} \big] \big( 1+(z_2-z_1) (A-z_2)^{-1} \big)
   \end{split}
  \end{equation*}
  holds; this can be seen by employing the first resolvent identity. Applying this for $A = (H[0]^\varepsilon_{\tilde\eta,\tilde\tau,\tilde\lambda}-i) \sigma_2$, $B=(H[0]_{\eta,\tau,\lambda}-i)\sigma_2$, $z_2=-k$, and $z_1=0$, we obtain
  \begin{multline*} 
      (H[k]^\varepsilon_{\tilde\eta,\tilde\tau,\tilde\lambda}-i)^{-1} - (H[k]_{\eta,\tau,\lambda}-i)^{-1} \\
      = \sigma_2 \Big(((H[0]^\varepsilon_{\tilde\eta,\tilde\tau,\tilde\lambda}-i)\sigma_2 + k)^{-1} - ((H[0]_{\eta,\tau,\lambda}-i)\sigma_2 + k)^{-1} \Big) \\
      = \sigma_2 \big( 1-k \big((H[0]_{\eta,\tau,\lambda}-i)\sigma_2+k\big)^{-1} \big) \big[ \sigma_2(H[0]^\varepsilon_{\tilde\eta,\tilde\tau,\tilde\lambda}-i)^{-1}-\sigma_2(H[0]_{\eta,\tau,\lambda}-i)^{-1} \big] \\
      \big( 1-k \big((H[0]^\varepsilon_{\tilde\eta,\tilde\tau,\tilde\lambda}-i) \sigma_2+k\big)^{-1} \big) \\
      = \big( \sigma_2-k (H[k]_{\eta,\tau,\lambda}-i)^{-1} \big) \sigma_2 \big[ (H[0]^\varepsilon_{\tilde\eta,\tilde\tau,\tilde\lambda}-i)^{-1}-(H[0]_{\eta,\tau,\lambda}-i)^{-1} \big] \sigma_2\\
      \big( \sigma_2-k (H[k]^\varepsilon_{\tilde\eta,\tilde\tau,\tilde\lambda}-i)^{-1} \big).
  \end{multline*}
 Since $\sigma_2$ is unitary and  $\|(A-i)^{-1}\|\leq 1$ holds for any self-adjoint operator $A$,   this implies
  \begin{equation} \label{equation_norm}
    \begin{split}
      \big\| (H[k]^\varepsilon_{\tilde\eta,\tilde\tau,\tilde\lambda}-i)^{-1} &- (H[k]_{\eta,\tau,\lambda}-i)^{-1} \big\| \\
      &\leq ( 1+|k| )^2 \big\| (H[0]^\varepsilon_{\tilde\eta,\tilde\tau,\tilde\lambda}-i)^{-1}-(H[0]_{\eta,\tau,\lambda}-i)^{-1} \big\|.
    \end{split}
  \end{equation}
For $K>0$, define $\mathscr{H}_K:=\{\psi \in L^2(\mathbb{R}^2; \mathbb{C}^2)|\, \psi(\cdot,k)=0 \text{ if }|k|>K\}$. 
First, take any $\psi\in\mathscr{H}_K$ and put $D_\varepsilon[k]:=(H[k]^\varepsilon_{\tilde\eta,\tilde\tau,\tilde\lambda}-i)^{-1} - (H[k]_{\eta,\tau,\lambda}-i)^{-1}$. Then we have
\begin{multline*}
\left\|\int_\R^\oplus D_\varepsilon[k]\dd k ~\psi\right\|^2=\left\|\int_{(-K,K)}^\oplus D_\varepsilon[k]\psi(\cdot,k)\dd k\right\|^2
=\int_{(-K,K)}\|D_\varepsilon[k]\psi(\cdot,k)\|^2\dd k\\
\leq \sup_{(-K,K)}\|D_\varepsilon[k]\|^2\int_{(-K,K)}\|\psi(\cdot,k)\|^2\dd k
=\sup_{(-K,K)}\|D_\varepsilon[k]\|^2\|\psi\|^2.
\end{multline*}
Using \eqref{equation_norm} we obtain the estimate 
\begin{equation} \label{eq:D_bound}
\left\|\int_\R^\oplus D_\varepsilon[k]\dd k ~\psi\right\|\leq (1+K)^2\|D_\varepsilon[0]\|\|\psi\|.
\end{equation}

Next, for a fixed $\psi=\int_\R^\oplus\psi(\cdot,k)\dd k\in L^2(\mathbb{R}^2; \mathbb{C}^2)$ and $K>0$ define the function $\psi_K:=\int_{(-K,K)}\psi(\cdot,k)\dd k\in\mathscr{H}_K$. Using the dominated convergence theorem, one can show that 
\begin{equation*}
\lim_{K\to\infty}\|\psi-\psi_K\|=0.
\end{equation*}
Finally, by the triangle inequality, \eqref{eq:D_bound}, and the fact that $\|\psi_K\|\leq\|\psi\|$, we get
\begin{multline} \label{eq:telescopic}
\left\|\int_\R^\oplus D_\varepsilon[k]\dd k ~\psi\right\|
\leq \left\|\int_\R^\oplus D_\varepsilon[k]\dd k ~\psi_K\right\|+\left\|\int_\R^\oplus D_\varepsilon[k]\dd k ~(\psi-\psi_K)\right\|\\
\leq (1+K)^2\|D_\varepsilon[0]\|\|\psi\|+\left\|\int_\R^\oplus D_\varepsilon[k]\dd k\right\| ~\|(\psi-\psi_K)\big\|.
\end{multline}
Since
$$\left\|\int_\R^\oplus D_\varepsilon[k]\dd k\right\|\leq\left\|\int_\R^\oplus (H[k]^\varepsilon_{\tilde\eta,\tilde\tau,\tilde\lambda}-i)^{-1}\dd k\right\|+\left\|\int_\R^\oplus (H[k]_{\eta,\tau,\lambda}-i)^{-1}\dd k\right\|\leq 2,$$
the second term on the right-hand side of \eqref{eq:telescopic} can be made arbitrarily small by choosing $K$ large enough.
With any fixed $K$ the first term on the right-hand side of \eqref{eq:telescopic} tends to zero as $\varepsilon \to 0$. We conclude that 
$\lim_{\varepsilon\to 0}\|\int_\R^\oplus D_\varepsilon[k]\dd k ~\psi\|=0$.
\end{proof}

\subsection*{Acknowledgements}
J. Behrndt and M. Holzmann gratefully acknowledge financial support by the Austrian Science Fund (FWF): P 33568-N. M.~Tu\v{s}ek was partially supported by the grant No.~21-07129S of the Czech Science Foundation (GA\v{C}R) and by the project CZ.02.1.01/0.0/0.0/16\_019/0000778 from the European Regional Development Fund. This publication is based upon work from COST Action CA 18232 MAT-DYN-NET, supported by COST (European Cooperation in Science and Technology), www.cost.eu. The authors wish to express their thanks to V\'{i}t Jakubsk\'{y} for pointing out several useful references.

\appendix
\section{Direct integrals of self-adjoint operators} \label{section_appendix}

In this section we will prove several abstract results concerning the direct integral of self-adjoint operators. 
First, following \cite[Section~XIII.16]{RS4} we recall some necessary notations.
Let $\mathscr{G}$ be a separable complex Hilbert space with inner product $\langle\cdot, \cdot \rangle_\mathscr{G}$ and, for simplicity, let $\mathcal{M} \subset \mathbb{R}$ be an interval. Moreover, let $A(k)$, $k \in \mathcal{M}$, be a family of self-adjoint operators in $\mathscr{G}$ such that $\mathcal{M} \ni k \mapsto \langle(A(k)-i)^{-1}\psi,\varphi\rangle_\mathscr{G}$ is measurable for all $\psi,\varphi \in \mathscr{G}$, and define in the Hilbert space
\begin{align*}
 \mathscr{H} & := \int_\mathcal{M}^\oplus \mathscr{G}~ \textup{d} k\\
 &  := \left\{ \psi: \mathcal{M} \rightarrow \mathscr{G}|\, k \mapsto \|\psi(k)\|_{\mathscr{G}} \text{ is measurable, } \|\psi\|^2_{\mathscr{H}}:=\int_\mathcal{M} \| \psi(k) \|_{\mathscr{G}}^2~ \textup{d} k < \infty \right\}
\end{align*}
the operator 
\begin{equation} \label{direct_integral}
  A := \int_\mathcal{M}^\oplus A(k) \textup{d} k
\end{equation}
by
\begin{equation*}
  \begin{split}
  (A\psi)(k) &= A(k) \psi(k),\\
    \dom (A) &= \left\{ \psi \in \mathscr{H}: \psi(k) \in 
    \dom (A(k)) \text{ a.e., } \int_\mathcal{M} \|A(k) \psi(k)\|_{\mathscr{G}}^2~\dd k<\infty \right\}.
  \end{split}
\end{equation*}
It is well-known that $A$ is self-adjoint in $\mathscr{H}$ and that
\begin{equation} \label{spectrum_direct_integral}
  \sigma(A) = \big\{ z \in \mathbb{R}|~\bigl|\{ k \in \mathcal{M}| \sigma(A(k)) \cap (z-\varepsilon, z+\varepsilon) \neq \emptyset\}\bigr| > 0\,\,\text{for all}\,\,\varepsilon>0  \big\}
\end{equation}
and
\begin{equation} \label{point_spectrum}
  \sigma_\textup{p}(A) = \big\{ z \in \mathbb{R}|~\bigl|\{ k \in \mathcal{M}| z \in \sigma_\textup{p}(A(k))\}\bigr| > 0 \big\},
\end{equation}
where $|B|$ denotes the Lebesgue measure of $B \subset \mathbb{R}$; cf. \cite[Theorem~XIII.85]{RS4}. Note that all eigenvalues of $A$ are of infinite multiplicity. Indeed, if $\lambda \in \sigma_\textup{p}(A)$, then for any normalized $\psi(k) \in \ker (A(k)-\lambda)$ and each function $\chi$ which is non-zero on $\{ k \in \mathcal{M}| z \in \sigma_\textup{p}(A(k))\}$, such that $k \mapsto \chi(k) \psi(k)$ is square integrable, the vector $k \mapsto \chi(k) \psi(k)$ belongs to $\ker(A-\lambda)$.

In the following theorem we provide a criterion implying that $\sigma(A)$ consists only of absolutely continuous spectrum and pure point spectrum. 
This criterion is used in Section~\ref{sec:2dOp} to analyze the spectrum of the operator $H_{\eta,\tau,\lambda}$. We assume for each $k \in \mathcal{M}$ that
\begin{equation*} 
\sigma_{\textup{sc}}(A(k))=\emptyset,
\end{equation*}
and that there exist at most countably many measurable sets $\mathcal{I}_n \subset \mathcal{M}$ and measurable  functions $E_n: \mathcal{I}_n \rightarrow \mathbb{R}$ such that for almost every $k \in \mathcal{M}$
\begin{equation*}
  \sigma_\textup{p}(A(k)) = \bigcup_{\{n|\, k \in \mathcal{I}_n \}} \{E_n(k)\}.
\end{equation*} 

The formulation of the following theorem has rather general assumptions on $E_n(k)$, afterwards in Corollary~\ref{corollary_ac_applicable} we discuss a special case which is easier applicable. 

\begin{theorem} \label{theorem_ac}
  Suppose that $A$ satisfies the above assumptions. Moreover, assume that for each $\mathcal{N} \subset \mathbb{R}$ with $|\mathcal{N}| = 0$ the relation $|E_n^{-1}(\mathcal{N} \setminus \sigma_\textup{p}(A))|=0$ holds for all $n$. Then
  \begin{equation*}
\sigma_\textup{sc}(A) = \emptyset.
  \end{equation*}
\end{theorem}
\begin{proof}
  Let $\mathscr{H}_\textup{pp}$ be the pure point subspace and let $\mathscr{H}_\textup{c} = \mathscr{H}_\textup{pp}^\perp$ be the continuous subspace associated with $A$.  
  It suffices to verify that $\mathscr{H}_\textup{c}$ is the absolutely continuous subspace associated with $A$. To show this, denote by $E$ and $E(k)$ the spectral measures corresponding to $A$ and $A(k)$, respectively. Let $\psi \in \mathscr{H}_\textup{c}$ be fixed. We check that for $\mathcal{N}\subset \mathbb{R}$ with Lebesgue measure zero
   the relation
  \begin{equation}\label{showthis}
     \int_\mathcal{N} \textup{d} \langle E\psi,\psi\rangle_{\mathscr{H}}  =0
  \end{equation}
  holds. In fact, we have
  \begin{equation*}
  \begin{split}
    \int_\mathcal{N} \textup{d} \langle E\psi,\psi\rangle_{\mathscr{H}} &= \int_\mathcal{M} \int_\mathcal{N} \textup{d}\langle E(k)\psi(k),\psi(k)\rangle_\mathscr{G}~ \textup{d} k\\
    &= \int_\mathcal{M} \int_{\mathcal{N} \setminus \sigma_\textup{p}(A)} \textup{d}\langle E(k)\psi(k),\psi(k)\rangle_\mathscr{G} \textup{d} k;
  \end{split}
  \end{equation*}
  cf. the proof of \cite[Theorem~XIII.85]{RS4} for the first equality and the second equality holds as 
  $\psi \in \mathscr{H}_\textup{c} = \mathscr{H}_\textup{pp}^\perp$  and $\sigma_\textup{p}(A)$ is at most countable.  By assumption, we can decompose $E(k) = E(k)_\textup{ac} + \sum_{n} P_n(k)$, where the measure $\langle E(k)_\textup{ac}\psi(k),\psi(k)\rangle_\mathscr{G}$ is absolutely continuous with respect to the Lebesgue measure and $P_n(k)$ is the orthogonal projection onto $\ker(A(k)-E_n(k))$. Then we have for each fixed $k \in \mathcal{M}$
  \begin{equation*}
    \int_{\mathcal{N} \setminus \sigma_\textup{p}(A)} \textup{d}\langle E(k)_\textup{ac} \psi(k),\psi(k)\rangle_\mathscr{G}  =0
  \end{equation*}
  and
  \begin{equation*}
    \int_{\mathcal{N} \setminus \sigma_\textup{p}(A)} \textup{d}\langle P_n(k) \psi(k),\psi(k)\rangle_\mathscr{G} \leq \begin{cases} \| \psi(k)\|^2_{\mathscr{G}} & \text{ for }  E_n(k) \in \mathcal{N} \setminus \sigma_\textup{p}(A), \\ 0 & \text{ for } E_n(k) \notin \mathcal{N} \setminus \sigma_\textup{p}(A). \end{cases}
  \end{equation*}
Let us denote by $\chi_{\mathcal{N} \setminus \sigma_\textup{p}(A)}$ the  indicator function associated with $\mathcal{N} \setminus \sigma_\textup{p}(A)$. Then we conclude that 
  \begin{equation*}
    \begin{split}
      \int_\mathcal{N} \textup{d} \langle E\psi,\psi\rangle_\mathscr{H} &\leq \int_\mathcal{M} \sum_{n} \| \psi(k)\|^2_{\mathscr{G}} \chi_{\mathcal{N} \setminus \sigma_\textup{p}(A)} (E_n(k))~ \textup{d} k \\
      &= \sum_{n} \int_{E_n^{-1}(\mathcal{N} \setminus \sigma_\textup{p}(A))}\| \psi(k) \|^2_{\mathscr{G}} ~\textup{d} k = 0,
    \end{split}
  \end{equation*}
  since $E_n^{-1}(\mathcal{N} \setminus \sigma_\textup{p}(A))$ is a zero set by assumption. This shows \eqref{showthis} 
  and yields the claimed result.
\end{proof}

In the following example we show that the assumption $|E_n^{-1}(\mathcal{N \setminus \sigma_\textup{p}(A)})|=0$ for all zero sets $\mathcal{N}$ is needed to conclude 
$\sigma_\textup{sc}(A) = \emptyset$.

\begin{example}
  Let
  \begin{equation*}
    \mathfrak{C} := \left\{ \sum_{n=1}^\infty \frac{a_n}{3^n} \bigg| a_n \in \{ 0,2\}~ \forall n \in \mathbb{N} \right\}
  \end{equation*}
  be the Cantor set. Choose for any fixed $x \in [0,1]$ \textit{one fixed} sequence of coefficients $\{ a_n(x): a_n(x) \in \{ 0,1\}, n \in \mathbb{N} \}$ such that
  \begin{equation*}
    x = \sum_{n=1}^\infty \frac{a_n(x)}{2^n}
  \end{equation*}
  and consider the map $E_1: [0,1] \rightarrow \mathfrak{C}$ given by
  \begin{equation*}
    E_1 \left( x \right) := \sum_{n=1}^\infty \frac{2 a_n(x)}{3^n}.
  \end{equation*}
  Then $E_1$ is strictly monotonously increasing, and hence injective and measurable. Consider in $L^2([0,1]) = \int_{[0,1]}^\oplus \mathbb{C}\, \textup{d} k$ the multiplication operator
  \begin{equation*}
    A = \int_{[0,1]}^\oplus A(k) \textup{d}k, \quad A(k) w = E_1(k)w \text{ for } k \in [0,1]\text{ and } w\in\C.
  \end{equation*}
  Clearly, $\sigma_\textup{c}(A(k))=\emptyset$, $\sigma(A(k))=\sigma_\textup{p}(A(k))=\{ E_1(k) \}$, and $\sigma_\textup{p}(A)=\emptyset$ follows from \eqref{point_spectrum}. However, for the spectral measure $E_A$ associated with $A$, $\psi \equiv 1$, and the Lebesgue zero set $\mathcal{N} = \mathfrak{C}$ we have
  \begin{equation*}
    \big\langle E_A(\mathfrak{C}) \psi, \psi \big\rangle_\mathbb{C} = \int_0^1 \chi_{E_1^{-1}(\mathfrak{C})}(k) \textup{d} k = 1 \neq 0,
  \end{equation*}
  i.e., $\psi$ belongs to the singularly continuous subspace of $A$ and hence, $\sigma_\textup{sc}(A) \neq \emptyset$.
\end{example}

Finally, we deduce the following result about the spectrum of a self-adjoint operator defined as a direct integral. The formulation is particularly simple to apply in the main part of the paper.

\begin{corollary} \label{corollary_ac_applicable}
Let the operator $A$ be defined via a direct integral as in~\eqref{direct_integral} such that for each $k \in \mathcal{M}$ 
\begin{equation*} 
\sigma_{\textup{sc}}(A(k))=\emptyset,
\end{equation*}
and assume that there exist at most countable many open intervals $\mathcal{I}_n \subset \mathcal{M}$ and real analytic  functions $E_n: \mathcal{I}_n \rightarrow \mathbb{R}$ such that for almost every $k \in \mathcal{M}$
\begin{equation*}
  \sigma_\textup{p}(A(k)) = \bigcup_{\{n|\, k \in \mathcal{I}_n \}} \{E_n(k)\}.
\end{equation*} 
Then $\sigma_\textup{sc}(A) = \emptyset$ 
and
the following assertions hold:
  \begin{itemize}
    \item[$\textup{(i)}$] $\sigma_\textup{p}(A) =  \bigcup_{n \in \mathcal{C}} \ran \, E_n$, where $\mathcal{C} = \{n|\,E_n\text{ is constant}\}$;
    \item[$\textup{(ii)}$] $\sigma(A)$ is given by 
    \begin{equation*}
    \overline{\bigcup_{n} \ran \, E_n}  \cup \bigl\{ z\in\mathbb R |\,\bigl|\{ k\in\mathcal M | \sigma_\textup{ac}(A(k)) \cap (z-\varepsilon, z+\varepsilon) \neq \emptyset\}\bigr| > 0\,\,\text{for all}\,\,\varepsilon>0  \bigr\};
    \end{equation*}
    \item[$\textup{(iii)}$] $\sigma_\textup{ac}(A)$ is given by 
    $$\overline{\bigcup_{n \notin \mathcal{C}} \ran \, E_n} \cup \bigl\{ z\in\mathbb R |\,\bigl|\{ k\in\mathcal M | \sigma_\textup{ac}(A(k)) \cap (z-\varepsilon, z+\varepsilon) \neq \emptyset\}\bigr| > 0 \,\,\text{for all}\,\,\varepsilon>0 \bigr\},$$
    where $\mathcal C$ is as in (i).
  \end{itemize}
\end{corollary}

\begin{proof}
  It suffices to prove the statement $\sigma_\textup{sc}(A) = \emptyset$, as the remaining assertions (i)--(iii) then follow from \eqref{spectrum_direct_integral} and~\eqref{point_spectrum}.
  According to Theorem~\ref{theorem_ac} we have to verify $|E_n^{-1}(\mathcal{N} \setminus \sigma_\textup{p}(A))| = 0$ for each measurable set $\mathcal{N}$ with $|\mathcal{N}|=0$. This will be done in two steps.
  
  {\it Step 1:} First, assume that $|E_n'|>0$. This and the substitution theorem imply 
  \begin{equation*}
    \begin{split}
      0  =\big|\mathcal{N} \setminus \sigma_\textup{p}(A)\big|  &= \int_{(\mathcal{N} \setminus \sigma_\textup{p}(A)) \cap \ran E_n} \textup{d} x + \int_{(\mathcal{N} \setminus \sigma_\textup{p}(A)) \setminus \ran E_n} \textup{d} x \\
      &= \int_{E_n^{-1}(\mathcal{N} \setminus \sigma_\textup{p}(A))} |E_n'(k)| \, \textup{d} k + \int_{(\mathcal{N} \setminus \sigma_\textup{p}(A)) \setminus \ran E_n} \textup{d} x\\
      &= \int_{E_n^{-1}(\mathcal{N} \setminus \sigma_\textup{p}(A))} |E_n'(k)| \, \textup{d} k\geq 0.
    \end{split}
  \end{equation*}
  Since $|E_n'|>0$, the last displayed formula  can only be true if $|E_n^{-1}(\mathcal{N} \setminus \sigma_\textup{p}(A))|  = 0$, which shows the assertion.

  {\it Step 2:} 
  Assume now that $E_n$ is an arbitrary real analytic function and that $E_n$ is not constant, as otherwise $E_n$ yields a point in $\sigma_\textup{p}(A)$. 
  Then the zeros of $E_n'$ can only accumulate to the boundary of $\mathcal I_n$ and if $\mathcal J_l$ denote the open intervals in between the zeros of $E_n'$, then
  $\vert E_n'\upharpoonright_{\mathcal J_l}|>0$. Hence, it follows for any measurable set $\mathcal{N}$ with $|\mathcal{N}|=0$ with the argument of {\it Step~1} that
  \begin{equation*}
    |E_n^{-1}(\mathcal{N} \setminus \sigma_\textup{p}(A))| = \sum_{l=1}^\infty |(E_n\upharpoonright_{\mathcal J_l})^{-1}(\mathcal{N} \setminus \sigma_\textup{p}(A))| = 0,
  \end{equation*}
  which yields all claims.
\end{proof}

\end{document}